\documentclass[letterpaper, 11pt, onecolumn]{article}

\usepackage{indentfirst}
\usepackage{setspace}

\usepackage[authoryear]{natbib}  
\usepackage[colorlinks=true, citecolor=blue]{hyperref}  

\usepackage{tabu}
\usepackage{epsfig} % for postscript graphics files
\usepackage{amsmath}
\usepackage{amssymb}  % assumes amsmath package installed

\usepackage{enumitem}
\usepackage{pifont}
\usepackage[dvipsnames]{xcolor}
\usepackage{tikz}
\usetikzlibrary{arrows,shapes}
\usetikzlibrary{arrows.meta}
\usepackage[T1]{fontenc}
\usepackage{aecompl}

\usepackage{pgfplots}
\usepackage{algorithm2e}
\usepackage{varwidth}
\usepackage{soul}
\usepackage{verbatim}
\usepackage{booktabs}
\usepackage{float}

\usepackage{graphics} % for pdf, bitmapped graphics files
\usepackage{caption}
\usepackage{subcaption}
\usepackage{multirow}
\usepackage{setspace}
\usepackage{graphicx}

\usepackage{titling}
\usepackage[top=1in, bottom=1in, left=1in, right=1in]{geometry}
\bibpunct[, ]{(}{)}{;}{a}{}{~}

\usepackage{amsthm}
\usepackage[scientific-notation=true]{siunitx}
\usepackage{amsfonts}
\usepackage{dsfont}
\usepackage{placeins}
\definecolor{darkblue}{RGB}{0,101,204}
\definecolor{carorange}{RGB}{255,131,0}

\newtheorem{theorem}{Theorem}

\newtheorem{corollary}{Corollary}

\theoremstyle{definition}
\newtheorem{definition}{Definition}
\newtheorem{assumption}{Assumption}
\theoremstyle{remark}
\newtheorem{remark}{Remark}

\newcounter{tmp}

\title{\Large\bf
When Altruism Meets Autonomy: Managing Bottleneck Congestion with Strategic Autonomous Vehicles}

\author{
Kexin Wang$^{1,*}$, 
Haohui He$^{1}$, 
Ruolin Li$^{1}$ \\[0.5em]
{\small \textit{$^{1}$Department of Civil and Environmental Engineering, University of Southern California}} \\[0.3em]
{\small \texttt{\{kwang255, haohuihe, ruolinl\}@usc.edu}}
}

\date{}

\begin{document}
\maketitle

\iffalse
\title{\LARGE \bf
When Altruism Meets Autonomy: Managing Bottleneck Congestion with Strategic Autonomous Vehicles}
\date{}
% <-this % stops a space
\maketitle
\vspace{-2.5cm} 
\begin{center}
  {\large Kexin Wang$^{1}$,\; Haohui He$^{1}$,\; and Ruolin {Li}^{*1}}\\[1em]
  {\small $^{1}$\textit{Department of Civil and Environmental Engineering, University of Southern California}\\[0.6em]
  \texttt{kwang255@usc.edu, haohuihe@usc.edu, ruolinl@usc.edu}}
\end{center}

\thispagestyle{empty}
\fi

%%%%%%%%%%%%%%%%%%%%%%%%%%%%%%%%%%%%%%%%%%%%%%%%%%%%%%%%%%%%%%%%%%%%%%%%%%%%%%%

\begin{abstract}

\noindent Weaving ramps are critical bottlenecks in highway networks due to conflicting traffic flows and complex interactions among heterogeneous vehicle types. In mixed-autonomy settings, the presence of controllable autonomous vehicles (AVs) introduces new opportunities to influence system-level outcomes, yet the \emph{structural} impact of such control remains poorly understood.
This paper develops a unified equilibrium framework to capture, predict, and optimize aggregate lane-choice behavior in weaving ramps with heterogeneous vehicle populations. We first formulate a Wardrop-based model capturing the selfish behavior of human-driven vehicles (HDVs) and establish existence, uniqueness and validity of the resulting equilibrium. We then introduce a Stackelberg–Wardrop formulation in which AVs act as strategic leaders optimizing system performance, while HDVs respond through equilibrium adaptation. The framework is further generalized to incorporate HDVs' and AVs' heterogeneous behavioral preferences via a Social Value Orientation (SVO) model.
Our analysis reveals a fundamental structural property of mixed-autonomy traffic systems: under selfish HDV behavior, the impact of AV penetration is inherently non-increasing, exhibiting plateau regions where performance remains unchanged and improves only at critical thresholds.
These results provide principled guidance for the design of AV control and incentive mechanisms in the presence of selfish human behavior, and demonstrate how strategically controlled autonomous agents can be deployed to induce system-level efficiency gains in mixed-autonomy transportation networks.

\noindent \textbf{Keywords:} Traffic equilibrium, Bottleneck congestion, Mixed autonomy, Game theory, Altruistic behavior
%Mixed-autonomy traffic, Weaving ramps, Lane choice, Wardrop equilibrium, Stackelberg game, Social value orientation, Autonomous vehicles, Altruism
\end{abstract}

\section{Introduction}\label{intro}

\noindent Mixed-autonomy traffic systems, in which AVs coexist with HDVs, are expected to fundamentally reshape transportation networks~\citep{talebpour2016influence,wu2017emergent}. A key distinction between these vehicle types lies in their decision-making capabilities: while HDVs act selfishly based on local information, AVs can be centrally coordinated and strategically controlled~\citep{zhang2018mitigating, ao2024control}. This asymmetry creates new opportunities to influence traffic equilibria and improve system performance. However, a fundamental question remains: how does the introduction of strategically controlled agents alter equilibrium outcomes in systems dominated by selfish users? Addressing this question is a central challenge in transportation science and the focus of this paper. Highway weaving ramps provide a representative bottleneck setting in which such interactions are especially pronounced~\citep{marczak2015macroscopic,chen2018capacity,lee2008empirical}. Motivated by this setting, we study the strategic interplay between HDVs and AVs at weaving ramp bottlenecks, and examine how these interactions can inform the design of AV control policies. Next, we motivate our problem setting from two complementary perspectives, mixed-autonomy systems and bottleneck characteristics at highway weaving ramps.

\subsection{Challenge of Mixed-Autonomy Systems}

\noindent Mixed-autonomy systems, characterized by the coexistence of AVs and HDVs, introduce a fundamental shift in how traffic systems operate. Unlike traditional settings composed solely of human drivers~\citep{beckmann1956studies,roughgarden2002bad}, these systems feature a heterogeneous population with distinct decision-making mechanisms: HDVs act selfishly based on local information, while AVs can be algorithmically controlled and coordinated at the system level~\citep{chen2017optimal}. This asymmetry creates the potential to influence traffic equilibria through the strategic deployment and control of AVs.

Existing studies have demonstrated that AVs can improve traffic efficiency, enhance stability, and enable new forms of coordinated control, even at moderate penetration levels~\citep{vinitsky2018benchmarks,wu2023leveraging}. These benefits arise from the ability of AVs to implement system-level objectives, such as throughput maximization and congestion mitigation, which are typically beyond the scope of selfish human drivers. As a result, AVs can act as controllable agents that reshape traffic dynamics and influence the behavior of surrounding vehicles.

However, the effectiveness of such control critically depends on the interaction between AVs and HDVs~\citep{lazar2017capacity, guo2022managing}. When AV penetration is low, the system remains dominated by selfish HDV behavior, which may limit or even neutralize the impact of strategically controlled AVs~\citep{lazar2017price}. As penetration increases, AVs gain greater influence over the aggregate traffic state, enabling more effective coordination and system-level improvements~\citep{lioris2017platoons, guo2020leveraging}. This suggests that the relationship between AV penetration and system performance is inherently nontrivial and mediated by equilibrium responses of HDVs, and also outlines a remaining fundamental gap: existing work largely focuses on performance improvements under specific control policies or simulation settings~\citep{vinitsky2018benchmarks, zhang2018mitigating}, but lacks a structural understanding of how strategically controlled AVs alter equilibrium outcomes in the presence of selfish users. In particular, it remains unclear how system performance evolves as AV penetration increases, and how behavioral heterogeneity shapes these dynamics.
Addressing this gap requires a unified framework that captures the strategic interaction between controllable AVs and selfish HDVs at equilibrium. This motivates the development of the equilibrium-based models proposed in this paper.

\subsection{Challenge of Highway Weaving Ramp}

\begin{figure*}[h!]
\centering
\includegraphics[width = 0.8\textwidth]{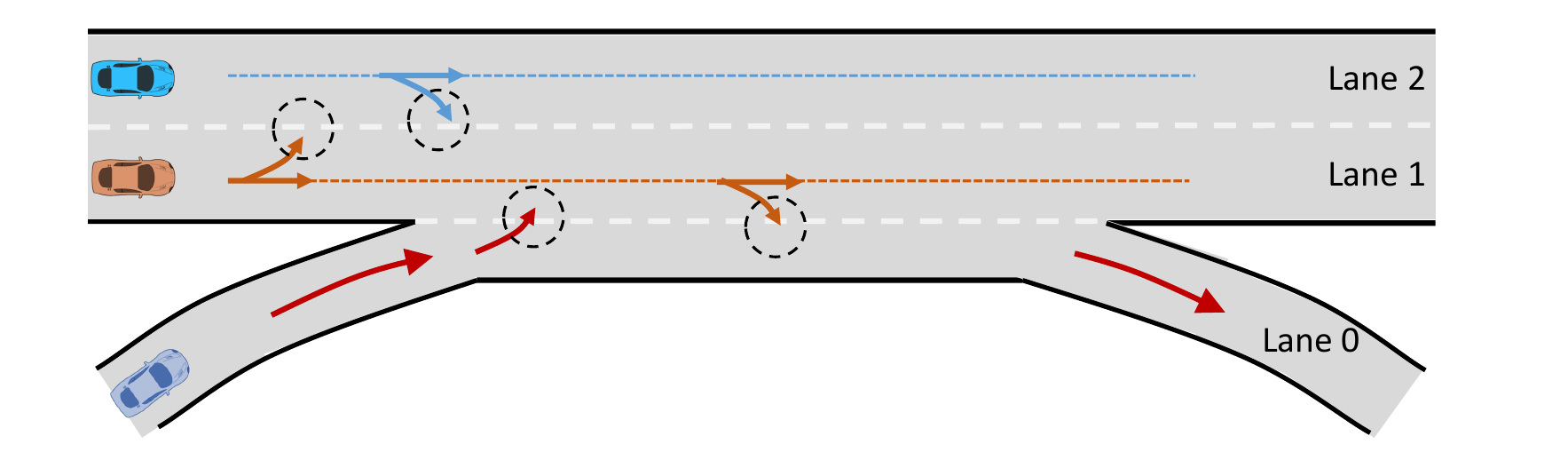}
\caption{A highway weaving ramp example, which shows complex interactions between entering, exiting and going through vehicles.}
\label{fig:interaction}
\end{figure*}

\noindent The interaction between AVs and HDVs becomes particularly pronounced in localized bottleneck settings, where vehicles must make rapid and strategic decisions under strong spatial and temporal constraints. Examples include intersections~\citep{dresner2008multiagent}, roundabouts~\citep{rodrigues2018autonomous}, and highway merging, diverging, and weaving segments~\citep{chandra2022gameplan}. In such environments, vehicles continuously adjust their behaviors, including merging, yielding, and changing lanes, in response to surrounding traffic conditions, leading to complex and highly coupled interactions.

Among these scenarios, highway weaving ramps provide a representative and analytically tractable bottleneck setting. In Figure.~\ref{fig:interaction}, it shows entering, exiting, and through vehicles must complete conflicting maneuvers within a limited roadway segment in this scenario. This induces strong competition for space and priority: entering vehicles merge into the mainline, exiting vehicles traverse across lanes to reach the off-ramp, and through vehicles decide whether to remain in their current lane or bypass congestion by changing lanes. These interdependent decisions generate substantial friction in traffic flow, leading to capacity reductions and congestion~\citep{laval2006lane}.

From a modeling perspective, weaving ramps naturally give rise to a strategic lane-choice problem~\citep{marczak2015macroscopic,lee2008empirical,ruolin2020onramp}. In particular, a subset of vehicles, namely those traveling through the weaving segment, face a fundamental trade-off between remaining in a congested lane or switching lanes to avoid conflicts, while incurring maneuvering costs~\citep{toledo2003lanechanging,kita1999merging,kita2002game}. This trade-off is inherently strategic, as the cost of each decision depends on the aggregate behavior of other vehicles. Moreover, in mixed-autonomy settings, AVs can influence these decisions through coordinated control, while HDVs respond selfishly based on observed conditions~\citep{talebpour2016influence}.

These features make weaving ramps an ideal setting to study how strategically controlled AVs interact with selfish HDVs and how such interactions shape equilibrium outcomes. At the same time, the strong coupling between merging and lane-changing behaviors amplifies congestion externalities, making it particularly challenging to design control strategies that achieve system-level improvements. This motivates the development of an equilibrium-based framework to analyze and manage lane-choice behavior in mixed-autonomy weaving bottlenecks.

\subsection{Our contributions}

\noindent This paper addresses the above question through the study of aggregate lane-choice behavior in highway weaving ramps. We develop a macroscopic game-theoretic framework that integrates (i) selfish HDV behavior modeled via Wardrop equilibrium, (ii) strategic AV control modeled through a Stackelberg formulation, and (iii) heterogeneous behavioral preferences captured using a Social Value Orientation (SVO) model.

Our analysis reveals a fundamental structural property of mixed-autonomy traffic systems. As AV penetration increases, system performance is not strictly improving; instead, the social cost follows a piecewise-constant regime structure, remaining unchanged over intervals of penetration and decreasing only at critical thresholds. Each regime is governed by a unique active behavioral type, and behavioral heterogeneity reshapes these thresholds, thereby altering the timing and magnitude of system-level improvements.

These findings suggest that mixed-autonomy traffic systems should not be viewed as smooth transitions from human-driven to fully automated operation. Rather, they exhibit regime-switching behavior driven by the interaction between strategic control and heterogeneous preferences. This perspective provides a new lens for understanding congestion dynamics and informs the design of AV deployment and control policies.

The main contributions of this study are as follows:
\begin{itemize}
    \item \textbf{Equilibrium Framework for Mixed Autonomy:} We develop a unified equilibrium framework for lane-choice behavior in weaving ramps, which integrates Wardrop equilibrium for HDVs with Stackelberg-based control of AVs, and further extends to heterogeneous populations through SVO-based preferences.

    \item \textbf{Existence, Uniqueness, and Structural Characterization:} We establish existence and uniqueness of the baseline selfish HDV equilibrium and characterize how mixed-autonomy equilibrium outcomes depend on traffic composition and control inputs.

    %\item \textbf{Threshold and Regime Analysis:} We derive explicit penetration thresholds that partition system behavior into distinct regimes and show that social cost exhibits a piecewise-constant structure with threshold-driven improvements.

    \item \textbf{Implications for AV Control Design:} We provide actionable insights for the design of AV control and incentive mechanisms, including conditions under which AV deployment leads to measurable reductions in system-level delay.
\end{itemize}

\section{Related Works}\label{related_works}

\noindent We next review the relevant literature along two key dimensions: models of lane-changing behavior and game-theoretic frameworks for analyzing strategic interactions in mixed-autonomy systems.

\subsection{Lane-Changing Behavior in Mixed Autonomy}

\noindent Lane-changing is a critical maneuver in traffic flow, as improper lane changes can create voids that trigger stop-and-go oscillations or even lead to breakdowns near bottlenecks~\citep{laval2008microscopic}. Empirical evidence suggest that approximately 10\% of crashes are attributable to lane-changing behavior~\citep{ji2020review}, underscoring the need for improved models of lane choice. Over the past decades, a variety of microscopic lane-changing models have been developed to describe when and how drivers decide to change lanes. Early rule-based frameworks, such as Gipps' seminal model~\citep{gipps1986model}, formalized drivers' decision processes based on safe distance and collision risk criteria. Later studies refined gap-acceptance logic by distinguishing free, cooperative, and forced maneuvers depending on urgency and driver willingness to decelerate~\citep{hidas2002modelling}. These classical models, along with subsequent works~\citep{sun2010research,li2019extended, li2020game}, established the foundation for simulating both mandatory and discretionary lane changes under human driving conditions. 

With the advent of AVs, research has begun to examine how partial automation may influence lane-changing dynamics and be leveraged to improve traffic performance~\citep{ghiasi2017mixed}. Data-driven approaches, particularly reinforcement learning, have been applied to train autonomous agents to replicate or optimize human-like lane-change decisions~\citep{shalev2016safe,wang2019lane}. Notably, studies show that even a small fraction of AVs can affect traffic conditions. For example, with 5-10\% AV penetration, properly designed AV control policies were found to increase average speeds and suppress stop-and-go waves in simulations~\citep{wang2020controllability}. Field experiments further demonstrate the potential of cooperative lane-changing~\citep{du2020cooperative,fan2025integrating}, where prototype AVs successfully executed coordinated lane maneuvers in live traffic, assisting HDVs during merges~\citep{adebisi2020developing}. Other system-wide analyses similarly confirm that AV deployment can improve efficiency and stability~\citep{li2018managing, guo2021mixed}. These findings illustrate that autonomy and V2X connectivity have powerful potential to reduce disturbances and enhance throughput during lane changes. 

Furthermore, highway weaving ramps present a particular challenging context for lane-changing in mixed autonomy. Weaving ramps are recognized as capacity bottlenecks, with throughput reductions of 3-20\% once congestion sets~\citep{yan2024bi}. Traditional studies attribute this "capacity drop" to human behavioral factors, including suboptimal gap acceptance and the spatial distribution of lane changes~\citep{cassidy1999some,treiber2013traffic}. AVs, by contrast, create opportunities for predictive and cooperative maneuvering: equipped with connectivity, they can act on a system level, unlike HDVs limited by perception and independent decision-making~\citep{talebpour2016influence,stern2018dissipation}. Recent research has thus proposed coordinated lane-changing control for weaving ramps, where centralized or distributed algorithms optimize lane-change trajectories across multiple vehicles~\citep{rios2016survey}. While promising in fully automated scenarios, these methods often require solving large-scale optimization problems and encounter scalability issues as vehicle numbers increase~\citep{zhou2016impact,chen2017towards}. Moreover, many assume high AV penetration and overlook the ongoing transition phase with mixed fleets~\citep{shladover2012impacts}. In summary, existing studies have largely focused either on localized, lane-changing decisions or on centralized coordination under idealized conditions, leaving a gap in modeling aggregate lane-choice behavior in weaving ramps that accounts for the strategic interplay between AVs and HDVs. To address this, our work introduces a set of game-theoretic formulations for lane choices in mixed-autonomy weaving zones.

\subsection{Game-Theoretic Modeling in Mixed Autonomy}

\noindent Game-theoretic frameworks have proven well-suited for studying mixed-autonomy environments, as they can capture interactive behaviors among heterogeneous agents within a unified decision-making model. By accommodating diverse objectives in a single framework, game theory enables a rigorous analysis of strategic decision-making in these complex settings. For instance,~\citet{li2023cooperative} formulated a Nash equilibrium approach to AV lane-changing that explicitly balances safety and energy efficiency. Similarly,~\citet{lopez2022game} combined normal-form game formulations, Nash equilibrium analysis, and Q-learning techniques to improve the stability of multi-agent lane-changing decisions. In another study,~\citet{talebpour2015modeling} employed a fictitious play-based framework to model AV lane-changing in a V2V communication-enabled environment, where drivers iteratively update their beliefs and compute best responses. Together, these studies demonstrate the effectiveness of game-theoretic methods for capturing realistic and strategic interactions among vehicles in mixed traffic.

Beyond modeling individual driver behaviors, game-theoretic principles also facilitate system-level traffic analysis. Wardrop's first principle formalizes the notion of a user-equilibrium traffic flow, wherein no driver can reduce their travel time by unilaterally changing behaviors~\citep{wardrop1952some}. While microscopic behavioral models are critical for operational-level decision making, aggregate equilibrium-based approaches provide valuable insights for high-level planning and policy design by transportation agencies and companies. In fact, such equilibrium conditions are often integrated into bi-level models~\citep{chiou2005bilevel}. These formulations often referred to as mathematical programs with equilibrium constraints, offer a principled framework for analyzing system-optimal strategies under rational user responses.

%In our prior work, we explored several game-theoretic frameworks to capture aggregate lane-choice behavior in complex traffic scenarios. For example, one study focused on a highway diverge with bifurcating lanes where vehicles must split into two lanes~\citep{li2019extended}, while another addressed lane choice at a highway on-ramp, where merging vehicles interact with the mainline traffic flow~\citep{li2020game,li2021employing}. These models were shown to effectively capture and predict HDVs’ lane-choice behavior in their respective scenarios. %Beyond these localized problems, we also extended our game-theoretic modeling to broader mixed-autonomy traffic networks by accounting for the unique operational characteristics of AVs. Specifically, we investigated the influence of AV headway policies on overall network-wide social delay~\citep{li2020impact}, designing toll lane frameworks that integrate AVs alongside high-occupancy vehicles~\citep{li2021highway}, and developing dynamic routing and queuing strategies under traffic-responsive intersection control schemes~\citep{li2022dynamic}. 
%Collectively, these studies illustrate the broad applicability of game-theoretic approaches, ranging from modeling macroscale lane-choice behavior to addressing system-level mixed-autonomy traffic management challenges.

In prior work, we developed several game-theoretic frameworks to capture aggregate lane-choice behavior in specific traffic scenarios. For example, one study considered a highway diverge with bifurcating lanes, where vehicles must split into two downstream links~\citep{li2019extended}, while another examined lane choice at highway on-ramps, where merging vehicles interact with mainline traffic~\citep{li2020game,li2021employing}. These models were shown to effectively capture and predict HDVs’ lane-choice behavior in their respective settings. Building on these foundations, the present work significantly extends the scope and depth of the analysis. We move beyond homogeneous, single-scenario models to a unified framework that incorporates mixed autonomy, strategic AV control, and heterogeneous behavioral preferences. This extension enables us to characterize how equilibrium structure evolves under the interaction between controllable and selfish agents, which is not captured in prior models.

%%%%%%%%%%%%%%%%%%%%%%%%%%%%%%%%%%%%%%%%%%%%%%%%%%%%%%%%%%%%%%%%%%%%%%%%%%%%%%%

\section{Problem Setting: Weaving Ramps}\label{problem_setting}

\noindent In this section, we introduce our core scenario: highway weaving ramps, defined as freeway segments where on-ramps and off-ramps are located in close proximity. In such settings, entering, exiting, and through vehicles perform intensive lane-changing maneuvers within a limited distance~\citep{HCM2010}. We first examine the inherent interaction complexity of such facilities, identify key lane-choice behaviors, and then define the essential notations for our framework. To maintain clarity, we focus on a representative case with two mainline lanes, as illustrated in Figure~\ref{fig:strategic}. 

\subsection{Key Strategic Lane Choice Behavior}

\begin{figure*}[h!]
\centering
\includegraphics[width=0.9\textwidth]{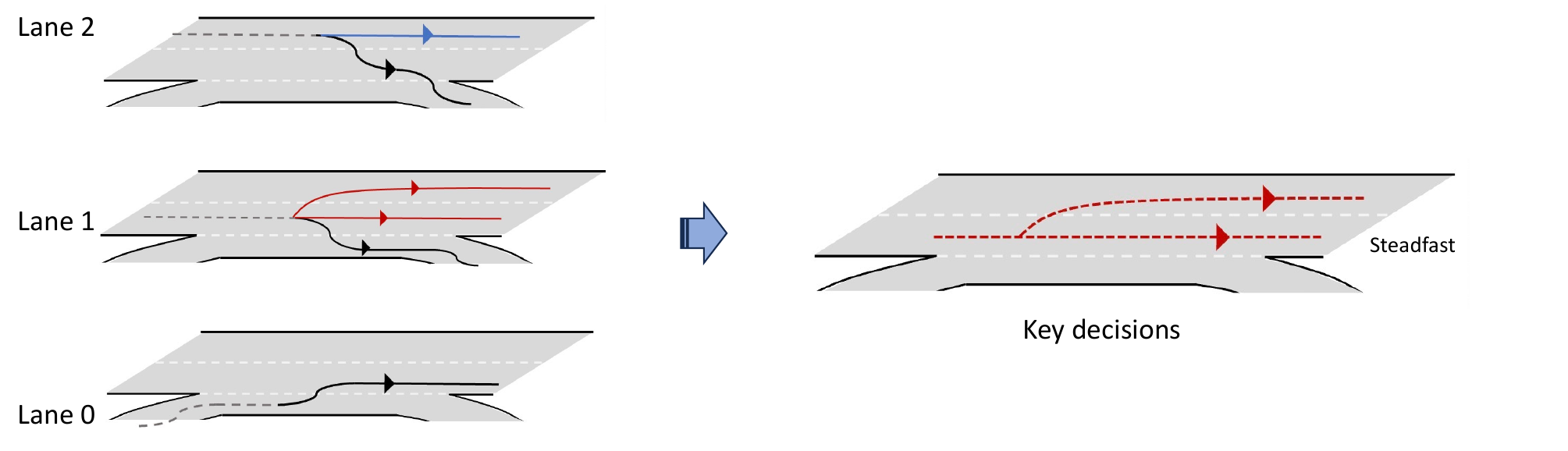}
    \caption{This figure illustrates the structure of lane-changing decisions within a highway weaving ramp. The flows shown in black and blue represent entering and exiting vehicles, along with through vehicles in the inner mainline lane (Lane 2), whose decisions are predetermined by their initial intentions or by the absence of lane-changing incentives, and are not subject to strategic adjustment. By contrast, the flows shown in red correspond to through vehicles in the outer mainline lane (Lane 1), which retain the flexibility to modify their lane-changing decisions in response to system-level delays and thus constitute the focus of our analysis.}
    \label{fig:strategic}
\end{figure*}

\noindent Most vehicles in this setting operate under predefined goal constraints. As illustrated in Figure~\ref{fig:strategic}, entering vehicles on Lane 0 must merge into Lane 1 to access the highway, while exiting vehicles on Lane 1 and 2 must leave via the weaving ramp. Through vehicles in Lane 2 have no incentive to change lanes. Consequently, their lane-changing decisions are relatively rigid and less responsive to external control strategies.

By contrast, through vehicles in Lane 1 represent a distinct class with more flexible objectives and greater susceptibility to external influence. These vehicles face two strategic decisions: (1) \textbf{steadfast in Lane 1}, accepting potential interactions and disruptions from both entering and exiting flows, or (2) \textbf{bypass to Lane 2} to avoid the conflict zone, while incurring the risks and time costs associated with lane changes. These decisions reflect a trade-off between exposure to disruption and the effort of maneuvering. 

Strategic decision-making also differs substantially between AVs and HDVs. HDV behavior is largely spontaneous and difficult to influence, particularly in short-horizon scenarios such as weaving ramps, where decisions must be made almost instantaneously. By contrast, AVs are more controllable, because they can access richer information and execute timely decisions. This controllability suggests that, as AV penetration increases in the near future, their decision-making can be leveraged to improve traffic system performance. Given this behavioral heterogeneity, our analysis focuses on how AVs make strategic decisions to regulate traffic flow and how HDVs respond selfishly to those decisions in Lane 1. In particular, we investigate how such interactions generate distinct behavioral patterns and evaluate their implications for both individual outcomes and system-level delays, which constitutes the central focus of this study.

\subsection{Notations}

\begin{figure*}[h!]
\centering
\includegraphics[width = 0.8\textwidth]{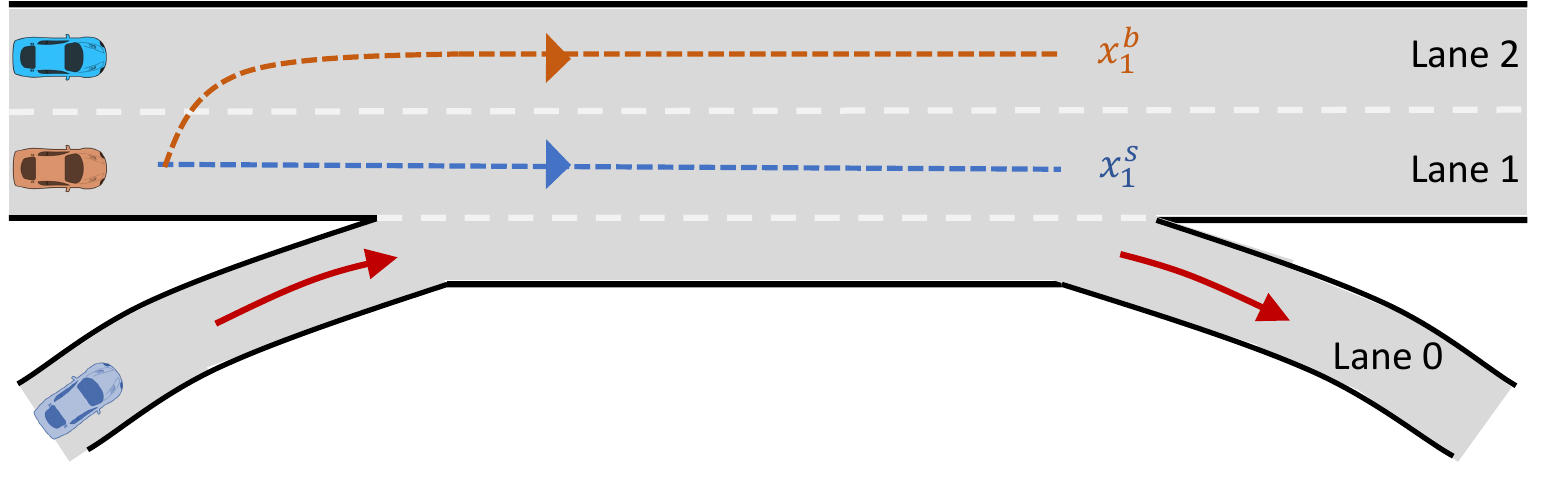}
\caption{The highway weaving ramp consists of three lanes. Lane 0 accommodates the entering flow $f_0^\text{enter}$, and Lane 2 carries both the exiting flow $f_2^\text{exit}$ and the non-strategic through flow $f_2^\text{s}$. These flows are treated as given parameters.
By contrast, Lane 1 contains two through flows $f_1^\text{s}$(steadfast) and $f_1^\text{b}$(bypass), which constitute the key decision variables in this study.}
\label{fig:geometry}
\end{figure*}

\noindent To evaluate the impact of strategic vehicle behavior on system performance and capture interactions among heterogeneous vehicle types, we adopt a macroscopic modeling perspective. Our focus is on the aggregate traffic flow patterns of each vehicle type, rather than the microscopic actions of individual vehicles. 

The vehicle types considered and their corresponding flow notations are as follows: (1) $f_0^{\text{enter}}$: entering vehicles on Lane 0 that merge into Lane 1; (2) $f_2^{\text{exit}}$: exiting vehicles on Lane 2 that traverse Lane 1 and exit the weaving ramp via Lane 0; (3) $f_2^\text{s}$: through vehicles on Lane 2; (4) $f_1^\text{s}$: through vehicles on Lane 1 that remain steadfast; (5) $f_1^\text{b}$: through vehicles on Lane 1 that bypass to Lane 2. Since our focus is on flows whose decisions can be strategically controlled, we distinguish between predetermined behaviors and strategic choices. Vehicle types (1)-(3) are goal-constrained, with actions either predetermined or lacking incentives for adjustment, and are therefore treated as \textbf{exogenous normalized flow ratios}. By contrast, vehicle types (4) and (5) represent decision-dependent variables whose allocations adapt to traffic conditions or control policies, and are modeled as \textbf{decision-dependent normalized flow ratios}.  

The \textit{exogenous normalized flow ratios} are defined as follows:
\begin{align}
  n_{0}^\text{enter} &:= \frac{f_{0}^\text{enter}}{f_{0}^\text{enter} + f_{2}^\text{exit} + f_{2}^\text{s}}, \\
  n_{2}^\text{exit} &:=\frac{f_{2}^\text{exit}}{f_{0}^\text{enter} + f_{2}^\text{exit} + f_{2}^\text{s}}, \\
  n_{2}^\text{s} &:=\frac{f_{2}^\text{s}}{f_{0}^\text{enter} + f_{2}^\text{exit} + f_{2}^\text{s}}, 
\end{align}
where we have $n_{0}^\text{enter} \ge 0, n_{2}^\text{exit} \ge 0, n_{2}^\text{s} \ge 0$, and naturally $n_{0}^\text{enter} + n_{2}^\text{exit} + n_{2}^\text{s} = 1$. The normalized flow ratios $n_0^{\text{enter}}$, $n_2^{\text{exit}}$ and $n_2^\text{s}$ represent the proportions of entering, exiting, and through vehicles relative to the total vehicle flow in the vicinity of Lane 1, originating from Lane 0 or Lane 2. These flow ratios are treated as known parameters in the model and are collectively denoted by the flow configuration vector $\mathbf{n} := \left(n_0^{\text{enter}},n_2^{\text{exit}},n_2^\text{s}\right).$

The \textit{decision-dependent normalized flow ratios} are defined separately as follows:
\begin{align}
    x_1^\text{s} := \frac{f_1^\text{s}}{f_1}, \\
    x_1^\text{b} := \frac{f_1^\text{b}}{f_1}, 
\end{align}
where naturally $x_1^\text{s} + x_1^\text{b} = 1$, $x_1^\text{s} \geq 0$, and $x_1^\text{b} \geq 0$. The variable \( x_1^\text{s} \) represents the proportion of steadfast vehicles on Lane 1, and \( x_1^\text{b} \) represents the proportion of bypassing vehicles on Lane 1. 

We incorporate AVs into the decision-dependent flows through the penetration rate $p$, as these flows are directly influenced by strategic decision-making. While AVs could in principle be included in the exogenous normalized flows, such an extension is unlikely to yield meaningful improvements in overall system performance. Accordingly, we restrict AVs to the decision-dependent categories and exclude them from $n_0^\text{enter}$, $n_2^\text{exit}$ and $n_2^\text{s}$. The total flow distribution vector is defined as $\mathbf{x}:=(x_{1,\text{CAV}}^\text{s},x_{1,\text{CAV}}^\text{b},x_{1,\text{HDV}}^\text{s},x_{1,\text{HDV}}^\text{b})$, where the aggregate decision-dependent flow ratios $x_1^\text{s}$ and $x_1^\text{b}$ are given by the sums of their AV and HDV components:
\begin{align}
   x_1^\text{s} = x_{1,\text{HDV}}^\text{s} + x_{1,\text{CAV}}^\text{s}, \label{eq:x_1^s} \\
   x_1^\text{b} = x_{1,\text{HDV}}^\text{b} + x_{1,\text{CAV}}^\text{b},
\end{align}
where we have $x_{1,\text{CAV}}^\text{s} + x_{1,\text{CAV}}^\text{b} = p$, and $x_{1,\text{HDV}}^\text{s} + x_{1,\text{HDV}}^\text{b} = 1-p$.

The relationships defined above establish the notation and the decomposition of flows into AV and HDV components. In the subsequent sections, we first formulate a strategic lane choice model under the Wardrop equilibrium framework for an HDV-only environment. This framework is then extended to a bi-level Stackelberg-Wardrop formulation, in which $\textbf{dedicated altruistic CAVs}$, which are centrally controlled AVs with prescribed strategies, are introduced to manage and influence system-level behavior. Finally, the model is generalized to a mixed-traffic setting that incorporates $\textbf{relaxed altruistic CAVs}$, i.e., AVs whose behavior is guided by incentivization mechanisms rather than direct control. All flow-ratio variables discussed in the remainder of the paper are defined with respect to this framework.

\iffalse
\begin{assumption}[On-ramp vehicles merge to Lane 1]
    In the weaving zone, on-ramp vehicles in Lane 0 aim to merge into Lane 1 within a constrained distance. okay if they bypass afterwards
\end{assumption}

\begin{remark}(feasibility of measurements)
    n2 and ne
\end{remark}

\begin{assumption}[Exit vehicles perform a series of lane changes to off ramp]
    doesn't impact if ne on lane 1
\end{assumption}
\fi

\iffalse
\begin{figure*}[h!]
    \centering
    \begin{subfigure}[b]{0.3\textwidth}
       \centering \includegraphics[width=\textwidth]{figs/JJ1.pdf}
        \caption{$J_1^\text{s}(x_1^\text{b})>J_1^\text{b}(x_1^\text{b})$}
    \end{subfigure}\qquad
    \begin{subfigure}[b]{0.3\textwidth}
       \centering \includegraphics[width=\textwidth]{figs/JJ2.pdf}
        \caption{$J_1^\text{s}(x_1^\text{b})$ and $J_1^\text{b}(x_1^\text{b})$ intersect}
    \end{subfigure} \qquad
    \begin{subfigure}[b]{0.3\textwidth}
       \centering \includegraphics[width=\textwidth]{figs/JJ3.pdf}
        \caption{$J_1^\text{s}(x_1^\text{b})<J_1^\text{b}(x_1^\text{b})$}
    \end{subfigure}  
    \caption{Three possible sketches of $J_1^\text{s}(x_1^\text{b})$ and $J_1^\text{b}(x_1^\text{b})$ in the region of $x_1^\text{b}\in [0,1]$.}
    \label{fig:JJ}
\end{figure*}
\fi

%%%%%%%%%%%%%%%%%%%%%%%%%%%%%%%%%%%%%%%%%%%%%%%%%%%%%%%%%%%%%%%%%%%%%%%%%%%%%%%%
\section{Selfish HDVs' Strategic Lane Choice Behavior as a Baseline} \label{sec:wardrop_eq_model}
%\section{Baseline: Homogeneous Human-Driven Scenario}

\noindent We begin by analyzing the decision-making process of HDVs in the absence of AVs, i.e., when AV penetration rate $p=0$ and ${x_{1,\text{CAV}}^\text{s}}={x_{1,\text{CAV}}^\text{b}}=0$. In this case, HDVs are assumed to behave selfishly, seeking to minimize their individual travel time delays while navigating the weaving ramp. Such behavior can be effectively captured by the Wardrop equilibrium framework~\citep{wardrop1952some}, which represents a state where no driver can unilaterally reduce travel time by changing lanes.

%Traditionally, traffic systems have been dominated by HDVs, resulting in a homogeneous traffic environment where each driver makes decisions independently to maximize personal efficiency. In weaving ramp scenarios, this self-interested behavior often leads to competitive lane-changing and suboptimal system outcomes. Such selfish behavior is well captured by the Wardrop equilibrium framework, which models user equilibrium as a state where no individual can reduce their travel time by unilaterally changing lanes. This formulation serves as a natural baseline for understanding lane-choice dynamics in HDV-only settings, and provides a benchmark against which more complex, mixed-autonomy scenarios can be evaluated.

\subsection{The Baseline Model Framework}

\noindent In the highway weaving ramps, decision-making vehicles choose between two strategic behaviors: \textbf{steadfast behavior} associated with a cost $J_1^\text{s}(\mathbf{x})$, and \textbf{bypassing behavior}, associated with a cost $J_1^\text{b}(\mathbf{x})$. The Wardrop-based model captures the equilibrium relationship between the steadfast flow rates $x_{1,\text{HDV}}^\text{s}$, the bypassing flow rate $x_{1,\text{HDV}}^\text{b}$, and their corresponding delay costs $J_1^\text{s}(\mathbf{x})$ and $J_1^\text{b}(\mathbf{x})$. At the lane choice equilibrium, no vehicles have incentive to change their lanes.

\begin{definition}[HDVs' Wardrop Lane Choice Equilibrium]\label{def:wdp_basic}
For a given weaving ramp configuration \( G = (\mathbf{N}, \mathbf{C}) \), a flow distribution vector \( \mathbf{x}:=(x_{1,\text{HDV}}^\text{s},x_{1,\text{HDV}}^\text{b}) \) is in equilibrium if and only if
\begin{subequations}\label{eq:eq_def}
    \begin{align}
        x_{1,\text{HDV}}^\text{s}  (J_1^\text{s}(\mathbf{x}) - J_1^\text{b}(\mathbf{x})) &\leq 0 ,\label{eq:eq_def_a}\\
        x_{1,\text{HDV}}^\text{b}  (J_1^\text{b}(\mathbf{x}) - J_1^\text{s}(\mathbf{x})) &\leq 0. \label{eq:eq_def_b}
    \end{align}
\end{subequations}
\end{definition}

This equilibrium model specifies that HDVs adopt the strategy that minimizes their own travel cost. Specifically, when bypassing incurs a lower cost (\( J_1^\text{s}(\mathbf{x}) - J_1^\text{b}(\mathbf{x}) > 0 \)), all Lane 1 through vehicles choose to bypass, resulting in \( x_{1,\text{HDV}}^\text{s} = 0 \). Conversely, when steadfast travel is less costly (\( J_1^\text{b}(\mathbf{x}) - J_1^\text{s}(\mathbf{x}) > 0 \)), all Lane 1 through vehicles remain in their lane, yielding \( x_{1,\text{HDV}}^\text{b} = 0 \). When the costs of the two strategies are equal (\( J_1^\text{b}(\mathbf{x}) - J_1^\text{s}(\mathbf{x}) = 0 \)), both behaviors may coexist, producing nonzero values of \( x_1^\text{s} \) and \( x_1^\text{b} \). The resulting Wardrop equilibrium represents a long-term, steady-state distribution of strategies, which may not occur at every moment in real traffic but captures the aggregate outcome of repeated behavioral adjustments over time.

\subsection{The Cost Model}\label{sec:wardrop_cost_func}

\begin{figure*}[h!]
\centering
\begin{subfigure}{0.25\textwidth}
    \includegraphics[width=\textwidth]{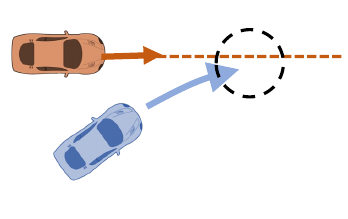}
    \caption{Merging Behavior}
    \label{fig:merging_behavior}
\end{subfigure}
\hspace{0.05\textwidth} 
\begin{subfigure}{0.25\textwidth}
    \includegraphics[width=\textwidth]{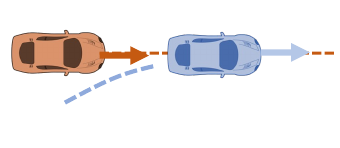}
    \caption{Traversing Behavior}
    \label{fig:traversing_behavior}
\end{subfigure}
\caption{Two types of traffic interactions contribute to the costs of bypassing and steadfast behaviors: merging and traversing. Merging occurs when a vehicle attempts to enter another lane, often leading to delays due to conflicts with existing traffic. Traversing refers to the condition where many vehicles remain in the same congested lane. These two behaviors are key components of the cost structure in the weaving segment.}
\end{figure*}

\noindent To model the costs associated with the two lane choice options ($steadfast$ and $bypass$) of through vehicles on Lane 1, some key conditions must be satisfied. First, from an aggregate perspective, through vehicles choosing the same option should have the same cost. Second, the cost should be an increasing function of the relevant traffic flows, reflecting the principle that higher traffic volumes lead to greater travel difficulty and, consequently, higher costs. More specifically, a cost function can consist of two components: a traversing cost, which is proportional to the total flow on the chosen lane, and a merging cost, which is proportional to the product of the flows of merging parties.

The specific structure and parameter values of the cost model are determined through calibration and validation using data collected from simulations. The model is considered appropriate only when the calibration and validation results demonstrate satisfactory accuracy. For the sake of brevity, we omit the iterative design process and present only the final validated models. Let \( J_1^\text{s} \) denote the cost for through vehicles choosing the steadfast option, and \( J_1^\text{b} \) denote the cost for through vehicles choosing the bypassing option. We have the following costs $J_1^\text{s}$ and $J_1^\text{b}$:
\begin{align}
J_1^\text{s}(\mathbf{x}) &= C_1^\text{t} \left( \alpha x_1^\text{s}+\beta n_2^\text{exit}+ n_0^\text{enter}\right) +C_1^\text{m} \left(\omega x_1^\text{s} n_2^\text{exit}+ x_1^\text{s} n_0^\text{enter}\right),\label{eq:Js}\\
J_1^\text{b}(\mathbf{x}) &= C_2^\text{t}\left(\gamma x_1^\text{b}+ n_2^\text{s}\right) + C_2^\text{m} \left( \rho x_1^\text{b} n_2^\text{s} + \delta x_1^\text{b} n_2^\text{exit} \right).\label{eq:Jb}
\end{align}

In the cost model, let \( C_i^\text{t} \) denote the unit traversing cost for Lane \( i \), where \( i \in \{1, 2\} \). After the strategic decisions of through vehicles, Lane 1 is shared by steadfast vehicles, exiting vehicles, and entering vehicles, while Lane 2 is occupied by bypassing vehicles and vehicles always traveling along Lane 2. The weight parameters \( \alpha, \beta, \gamma \) are assumed to be positive, representing the relative impact on the cost compared to a standard traversing behavior. If a parameter value exceeds 1, it indicates a higher cost relative to normal traversing, possibly due to lower speeds or increased discomfort. Conversely, values less than 1 suggest a lower impact, while a value of 1 indicates a neutral, baseline cost. Similarly, let \( C_i^\text{m} \) represent the unit merging cost for Lane \( i \), with \( \omega, \rho, \delta \) denoting the additional effort or discomfort associated with conflicting merging movements. A higher value of \( \delta \) implies a more challenging merging process, reflecting the increased difficulty compared to standard merging behavior. Then let the cost coefficient vector as \(\mathbf{C} := (C_i^\text{t}, C_i^\text{m}, \alpha, \beta, \omega, \gamma,\rho, \delta \, | \, i \in \{1, 2\})\), which includes all parameters that need be calibrated prior to the deployment on a new weaving ramp.

\begin{remark}[Application Scenarios of Our Model]
    The proposed cost functions, which are proportional to flow rates, are particularly well-suited for modeling uncongested scenarios where traffic flows smoothly. However, they may not fully capture the dynamics of congested conditions, such as queue formation, and thus are less applicable in oversaturated environments.
\end{remark}

\subsection{Existence and Uniqueness of HDVs' Wardrop Equilibrium}
\noindent The aggregate lane choice equilibrium implies that each vehicle selects the option that minimizes its own cost, resulting in one of three outcomes: (1) all vehicles remain steadfast, (2) all vehicles bypass, or (3) a mixture of both behaviors if the costs $J_1^\text{s}$ and $J_1^\text{b}$ are equal. We now establish the existence and uniqueness of this equilibrium, showing how HDVs' equilibrium strategies vary across different exogenous conditions. 

\begin{theorem}[Existence and Uniqueness]\label{thm:uniq}
For any given weaving ramp configuration \( G = (\mathbf{N}, \mathbf{C}) \), the equilibrium flow distribution vector \( \mathbf{x} \) defined in Definition~\ref{def:wdp_basic} always exists and is unique.
\end{theorem}

The proof details are provided in Appendix \ref{sec:appendix_hdv_eq}. This proof establishes the existence and uniqueness of the HDV equilibrium, ensuring that a lane-choice equilibrium is guaranteed to exist for any real weaving ramp. With this result, our Wardrop-based lane choice model provides a well-defined baseline for subsequent analysis. 

\subsection{Parameter Calibration and Model Validation}\label{4A}

\noindent In this section, we calibrate the parameters and validate the proposed Wardrop-based lane choice equilibrium model for HDVs. The data used for calibration are obtained from the microscopic traffic simulator SUMO~\citep{SUMO2012}. An overview of the simulated highway weaving ramp scenario is shown in Figure~\ref{fig:geometry}.

\subsubsection{Parameter Calibration}

\noindent To evaluate the performance of our model, we first calibrate the cost coefficient vector \(\mathbf{C}\) using the optimization method detailed in~\citep{mehr2021game,li2019extended}. A broad range of scenarios is selected for data generation to ensure a comprehensive understanding of the model. Specifically, the simulations are conducted under a total flow rate of 1400 vehicles per hour. In the total flow rate, 600 vehicles per hour is allocated among the three neighboring vehicle types $f_2^\text{s}, f_2^\text{exit}, f_0^\text{enter}$. The remaining flow rate of 800 vehicles per hour is set for $f_1^\text{s}$ and $f_1^\text{b}$. This ensures sufficient interaction without triggering excessive congestion that could cause a deadlock in the simulation. We further fix the flow rate of one of the three vehicle types \(f_0^\text{enter}\), \(f_2^\text{exit}\), and \(f_2^\text{s}\) while varying the flow rates of the other two. This results in a total of 415 distinct data points. Each data point captures the equilibrium flow distribution vector \(\mathbf{x} = (x_1^\text{s}, x_1^\text{b})\) for a simulation run with a given flow configuration (\(n_0^\text{enter}\), \(n_2^\text{exit}\), \(n_2^\text{s}\)). To ensure that each simulation run reaches a steady state, we set each simulation duration to 20,000 timesteps, with the length of a single timestep as 1 second. All the vehicles have been successfully loaded into the simulation. 

%\begin{remark}[Steady state is key to obtaining quality data]
%    When calculating the equilibrium flow distribution for each data point, we monitor the ratio of the number of vehicles making different decisions to the total number of through vehicles on Lane 1 accumulating in a time period. The simulation only reaches a steady state if the duration is sufficiently long. Based on our tests, a simulation duration of over 5,000 timesteps is required for the traffic to stabilize. Therefore, we set the duration to 20,000 timesteps, ensuring that the simulation results are as realistic and stable as possible to expose the equilibrium state while maintaining acceptable computational efficiency.
%\end{remark}

Solving the calibration optimization problem\citep{mehr2021game,li2019extended}, which aims to find the best parameters that enable as many data points as possible to satisfy the condition in Definition \ref{def:wdp_basic}, subject to the constraints on unit costs, we obtain:
\begin{align}
    C_1^\text{t} &= C_2^\text{t} = 1, \\
    C_1^\text{m} &= C_2^\text{m} = 1.
\end{align}
We then obtain the calibrated cost coefficients as follows:
\begin{align}
    \alpha &= 1.255, \quad
    \beta = 1.138, \quad
    \omega = 1.000, \\
    \gamma &= 2.384, \quad
    \delta = 3.094, \quad
    \rho = 1.000.
\end{align}

The above parameters reflect the relative impact of different vehicle interactions on the overall cost. Specifically, \(\alpha\) and \(\beta\) are above 1 and $\omega =1$, which suggests that the presence of steadfast vehicles in Lane 1 introduces additional discomfort or delay, possibly due to the increased conflicts and the need to brake and slow down. Further, lane-changing maneuvers of exiting vehicles disrupt the flow in Lane 1, increasing travel time or creating a need for additional adjustments by steadfast vehicles, while entering vehicles merge into Lane 1 without significantly disrupting the traffic flow, possibly due to the adequate space between the ramps leading to a neutral impact on steadfast vehicles. On lane 2, \(\gamma\) is significantly greater than 1, suggesting that the traversing cost for bypassing vehicles in Lane 2 is higher compared to the baseline traversing cost. This could be due to the increased discomfort of needing to maneuver around other vehicles, reflecting the high cost of making a bypass decision. And the higher value of \(\delta\) suggests that conflicting merging movements significantly increase the cost by over three times than the standard merging cost. This reflects the increased difficulty and safety risks associated with such complex maneuvers, making it a significant factor in the decision-making process for through vehicles facing the weaving zone.

\begin{remark}[Parameter values can reflect intrinsic characteristics of the weaving ramp]
    Although parameter magnitudes are determined through calibration, they are influenced by factors such as ramp geometry, speed limits, and other intrinsic features of the weaving zone. Hence, the calibrated values should be interpreted within the context of the specific scenario and may vary across different conditions. The effects of speed limits, vehicle gap, and driver aggressiveness on these parameters are analyzed in Appendix~\ref{sec:appendix_model_analysis}.
\end{remark}

%\subsection{Model Validation}\label{sec:model_validation}

\subsubsection{Model Validation}

\noindent While validating our model, we selected a broad range of data to demonstrate its predictive accuracy and robustness. Specifically, we used 320 distinct data points, covering the range of \( n_0^\text{enter},n_2^\text{exit},n_2^\text{s} \in [0.2, 0.8] \), separate from the data points used for calibration, while ensuring the total demand across these three vehicles types is fixed at 600 vehicles per hour, with no changes made to other parameters.

Fig.~\ref{fig:validation} presents the validation results, illustrating that the predictions of our lane choice model align closely with the observed simulation outcomes. Subfigures (a) and (c) indicate that, when the ratio of entering vehicles is fixed, increasing the proportion of exiting vehicles leads to more frequent interactions on Lane~1 and consequently more bypassing behavior by through vehicles, as they shift to Lane~2 to avoid congestion and reduced speeds. Subfigures (b) and (d) further illustrate that, when the normalized flow of through vehicles on Lane 2 is fixed, a higher entering ratio still results in more bypassing on Lane~1 than a higher exiting ratio. This finding indicates that entering vehicles exert a stronger influence on the bypassing behavior of through vehicles on Lane 1 than exiting vehicles. The phenomenon also corresponds to the high value of $\delta$ in the cost function, as exiting vehicles substantially increase the bypassing cost.

It is noteworthy that, despite the linear structure of our cost models, the lane choice model can still capture the nonlinear lane choice behavior observed in Fig.~\ref{fig:validation} due to the inherent nonlinearity of the inequalities governing the equilibrium conditions. As a result, our model demonstrates significant flexibility and minimal requirements for computation and calibration while satisfying accuracy.

\section{Is the Selfish Behavior Socially Optimal?}
\begin{figure*}[h!]
\centering
\begin{subfigure}{0.45\textwidth}
    \includegraphics[width=\textwidth]{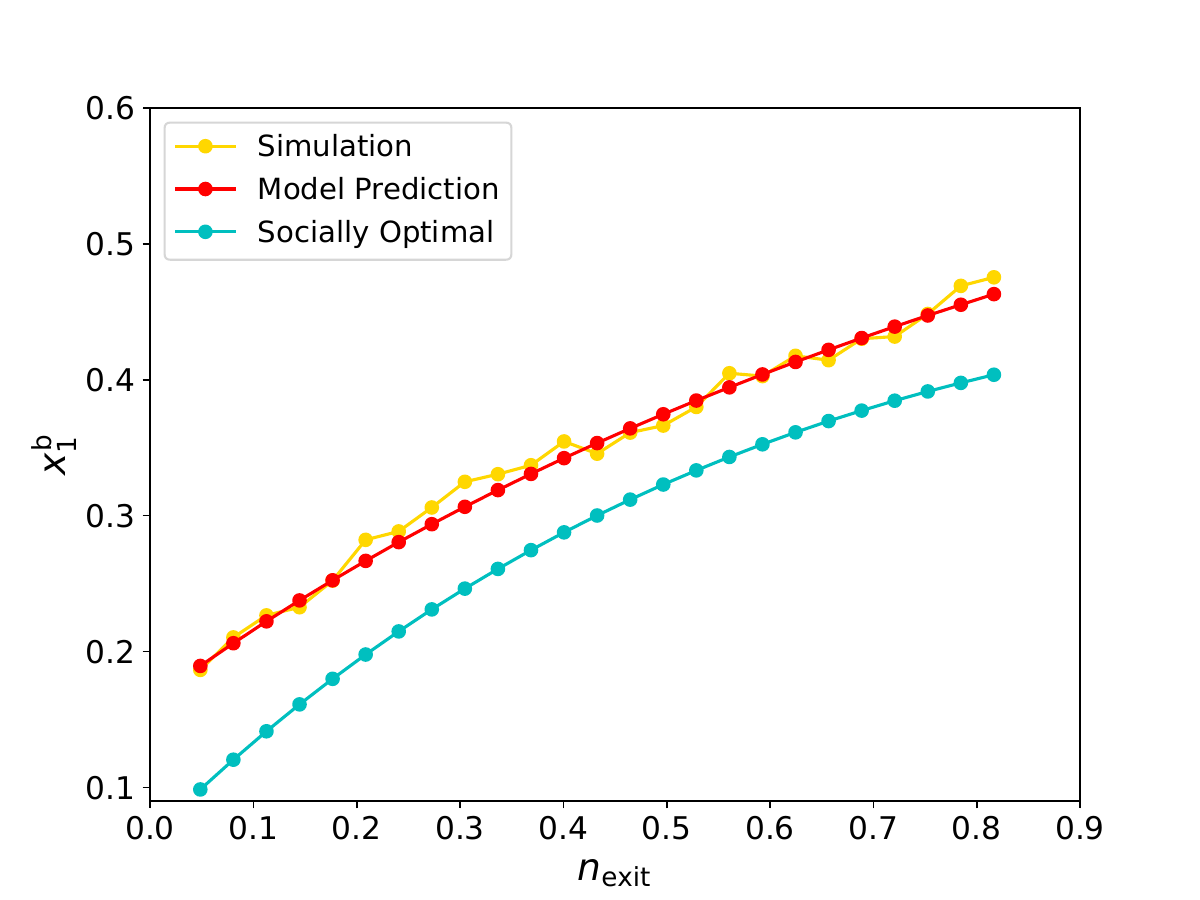}
    \caption{Validation result with $n_\text{enter}=0.1667$}
    \label{fig:first}
\end{subfigure}
\hfill
\begin{subfigure}{0.45\textwidth}
    \includegraphics[width=\textwidth]{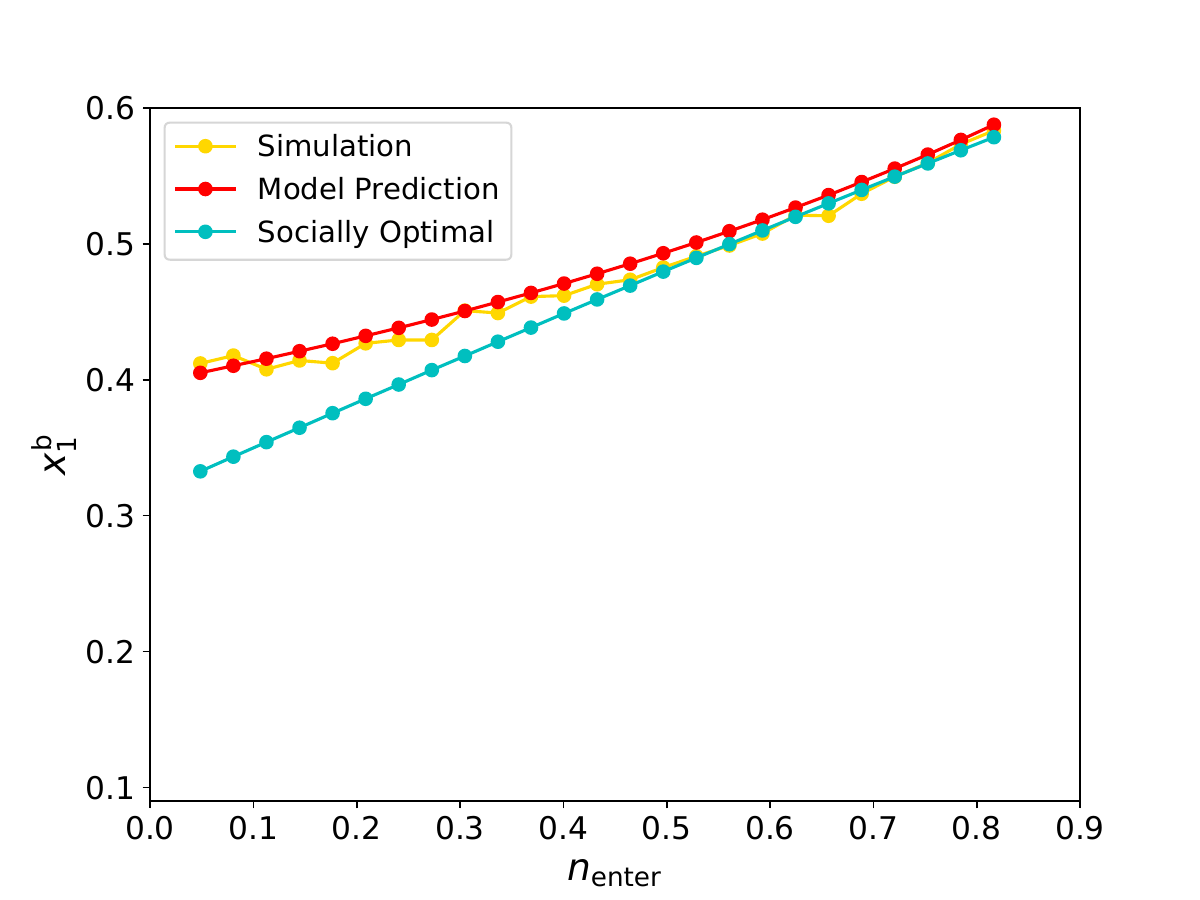}
    \caption{Validation result with $n_\text{2}=0.1667$}
    \label{fig:second}
\end{subfigure}

\vspace{0.5cm} % Adjust vertical space between rows if needed

\begin{subfigure}{0.45\textwidth}
    \includegraphics[width=\textwidth]{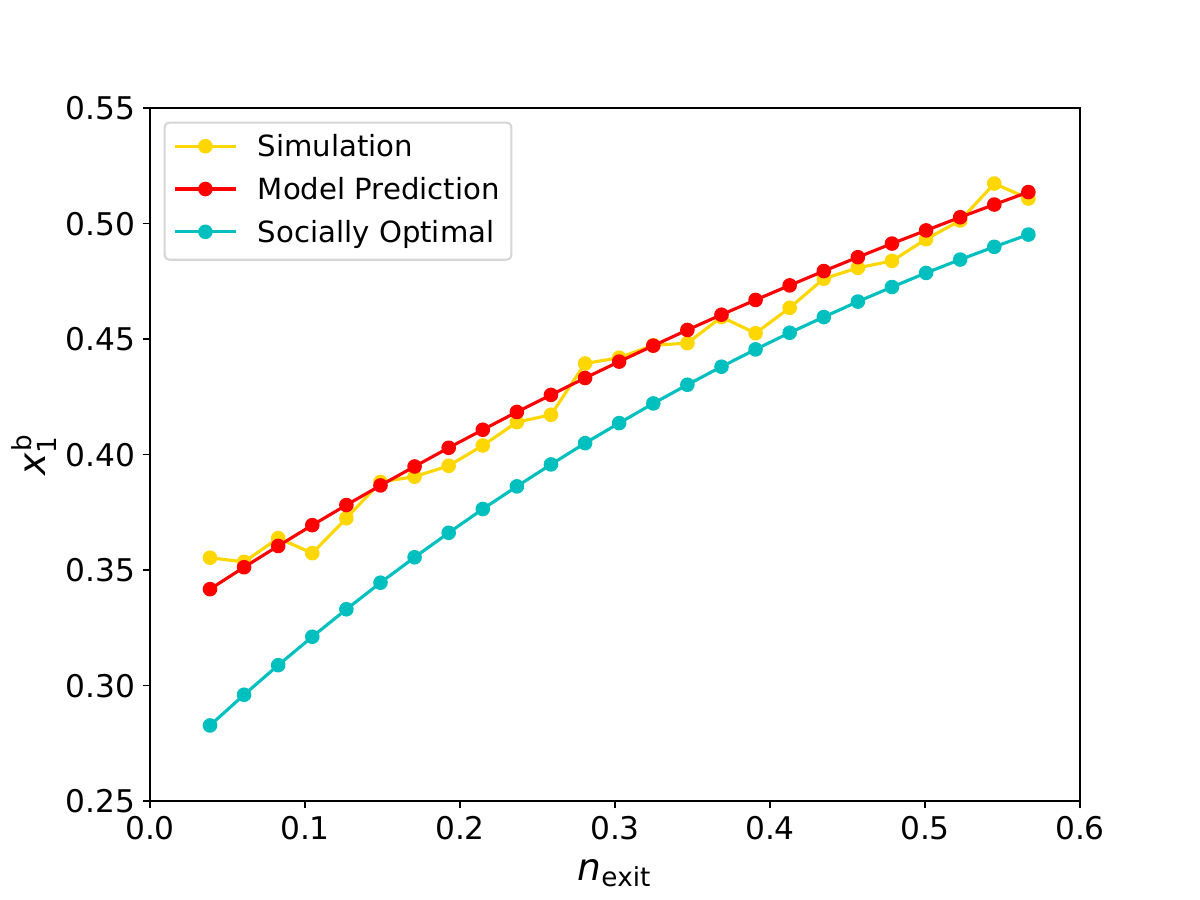}
    \caption{Validation result with $n_\text{enter}=0.4167$}
    \label{fig:third}
\end{subfigure}
\hfill
\begin{subfigure}{0.45\textwidth}
    \includegraphics[width=\textwidth]{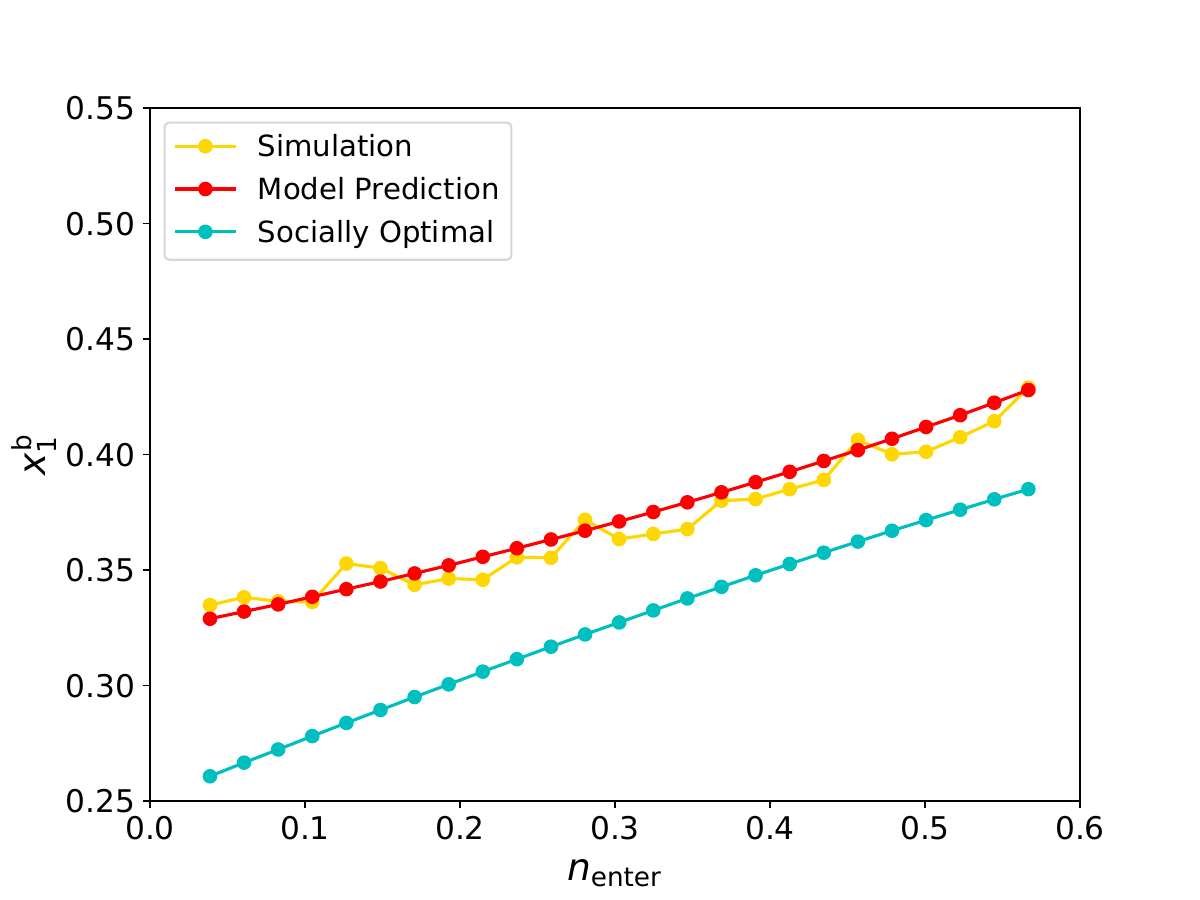}
    \caption{Validation result with $n_\text{2}=0.4167$}
    \label{fig:fourth}
\end{subfigure}
\caption{Validation results demonstrate strong agreement between the proposed lane-choice model and observed simulation outcomes. Despite the underlying cost model is linear, the lane choice model successfully captures nonlinear behavior patterns due to the inherent nonlinearity of the equilibrium conditions. Moreover, the validation consistently reveals a gap between the equilibrium outcomes of our lane-choice model and the socially optimal benchmark across all tested configurations.}
\label{fig:validation}
\end{figure*}

\noindent Building on the Wardrop-based model introduced above, we further compare the user equilibrium (UE) and the socially optimal (SO) outcomes in this section. We first define the social cost $J_\text{soc}$ as the total delay experienced by all vehicles within the weaving ramp, which is calculated as the sum of individual behavioral costs weighted by their corresponding flow magnitudes:
\begin{align} \label{eq:social_cost}
    J_\text{soc} (\textbf{x}) &= x_{1}^\text{s} J_{1}^\text{s} + x_{1}^\text{b} J_{1}^\text{b} + n_{2}^\text{s} J_{2}^\text{s} + n_{2}^\text{exit} J_{2}^\text{exit} + n_{0}^\text{enter} J_{0}^\text{enter},
\end{align}
where $J_1^\text{s}, J_1^\text{b}, J_2^\text{s},J_2^\text{exit}, J_0^\text{enter}$ denote the cost functions corresponding to steadfast behavior on Lane 1, bypassing behavior on Lane 1, steadfast behavior on Lane 2, exiting from Lane 2, and entering from Lane 1, respectively. The detailed description of $J_2^\text{s},J_2^\text{exit}$, and $J_0^\text{enter}$ is provided in Section \ref{subsec:updated_notation}. 

In Equation \eqref{eq:social_cost}, we already know the cost function $J_1^\text{s}$ and $J_1^\text{b}$ from Equation \eqref{eq:Js} and \eqref{eq:Jb}. By following the formulation logic and structure form for these functions, we extend the cost functions for entering, exiting and through behaviors as follows:
\begin{align}
  J_{2}^\text{s} (\textbf{x}) &= C_2^\text{t}(\gamma x_1^\text{b}+n_2^\text{s}) + C_2^\text{m}x_1^\text{b}n_2^\text{s}, \label{eq:J_2^s}\\
  J_{2}^\text{exit} (\textbf{x}) &= C_1^\text{t}(\alpha x_1^\text{s}+\beta n_2^\text{exit} + \omega n_0^\text{enter}) + C_1^\text{m}(x_1^\text{s}n_0^\text{enter} + x_1^\text{s}n_2^\text{exit}) + C_2^\text{m} \delta x_1^\text{b}n_2^\text{exit} ,\label{eq:J_2^exit} \\
  J_{0}^\text{enter} (\textbf{x}) &= C_1^\text{t}(\alpha x_1^\text{s}+\beta n_2^\text{exit} + \omega n_0^\text{enter}) + C_1^\text{m} (x_1^\text{s}n_0^\text{enter} + x_1^\text{s}n_2^\text{exit}),\label{eq:J_0^enter}
\end{align}
where $C_i^\text{t}$ denotes the unit traversing cost on Lane $i\in [1,2]$, $C_i^\text{m}$ denotes the unit merging cost on Lane $i\in [1,2]$, and the parameters $\alpha, \beta, \gamma, \omega, \delta$ capture the relative influence of different interactions on the overall cost. $J_2^\text{s}$ is determined by congestion within the lane and the merging interactions caused by bypassing vehicles entering it. $J_2^\text{exit}$ captures congestion arising from the traversing flow of steadfast, entering, and exiting vehicles, together with the merging difficulties as exiting vehicles sequentially merge into Lanes 1 and 0. $J_0^\text{enter}$ depends on congestion in Lane 1 and on the merging interactions between entering and exiting vehicles.

Equation~\eqref{eq:social_cost} is then used to compute the socially optimal solution, illustrated as the “socially optimal” line in Figure~\ref{fig:validation}. The numerical results show that HDVs’ selfish decisions consistently deviate from the socially optimal outcomes. Subfigures (a) and (d) demonstrate a persistent gap between the Wardrop-based model predictions and the social optimum across all variations of exogenous normalized flows under fixed conditions $n_0^\text{enter}=0.1667$ and $n_2^\text{s}=0.4167$. Subfigures (b) and (c) reveal a similar pattern, when $n_0^\text{enter}$ is small, the deviation between UE and SO is pronounced, but this gap gradually narrows as $n_0^\text{enter}$ increases. 

%From a theoretical perspective, following~\citet{beckmann1956studies}, Wardrop equilibrium can be formulated as the solution to a convex optimization problem, which shows equilibrium flows are generically inefficient relative to the social optimum.~\citet{roughgarden2002selfish} introduced that whenever congestion effects are present, selfish decisions are always yield higher system cost than the social optimum~\citep{roughgarden2002price,correa2008geometric}. These theoretical insights confirm that selfish HDV behavior is always leads to suboptimal system performance. 

Thus, achieving socially optimal outcomes is inherently difficult in HDV-only environments. To enhance traffic efficiency and mitigate this inefficiency, we introduce two types of CAVs, which are \textbf{dedicated altruistic CAVs} and \textbf{relaxed altruistic CAVs} as control mechanisms designed to improve overall system performance.

%Figure.~\ref{fig:validation} presents a comparison between the Wardrop-based model predictions, the socially optimal lane choices, and the SUMO simulation results. The numerical results illustrate that HDVs' selfish decisions consistently deviate from the socially optimal outcomes. Subfigures (a) and (d) demonstrate a persistent gap between the Wardrop-based model prediction and the social optimum across all variations of exogenous normalized flows, under the fixed conditions $n_0^\text{enter}=0.1667$ and $n_2^\text{s}=0.4167$. Subfigures (b) and (c) reveal a similar pattern, where a noticeable gap exists when $n_0^\text{enter}$ is small, and this gap gradually narrows as $n_0^\text{enter}$ increases. 

\section{AVs as Dedicated Altruistic Decision Makers} \label{sec:dictated_CAV}

\noindent The Wardrop-social optimality gap identified in the previous section highlights an opportunity for system-level intervention. To mitigate this inefficiency, we introduce dedicated altruistic CAVs, a class of AVs that follow externally mandated, system-oriented objectives. Such objectives can be communicated by a central authority through low-latency V2X links~\citep{talebpour2016influence}. Unlike HDVs, whose behavior is driven by individual self-interest, dedicated altruistic CAVs' strategic behavior is governed by prescribed system-level goals. This unique characteristic makes them a promising mechanism for traffic management and optimization via strategic behavioral control. In this section, we develop a Stackelberg-Wardrop framework to capture the hierarchical interactions between CAVs and HDVs, and further investigate how system-level delays respond to varying CAV penetration rates through numerical experiments.

\subsection{Updated Notations for Dedicated Altruistic CAVs}\label{subsec:updated_notation}

\iffalse
\begin{align}
J_1^\text{s}(\mathbf{x}) &= C_1^\text{t} \left( \alpha x_1^\text{s}+\beta n_2^\text{exit}+ n_0^\text{enter}\right) +C_1^\text{m} \left(\omega x_1^\text{s} n_2^\text{exit}+ x_1^\text{s} n_0^\text{enter}\right),\label{eq:Js}\\
J_1^\text{b}(\mathbf{x}) &= C_2^\text{t}\left(\gamma x_1^\text{b}+ n_2^\text{s}\right) + C_2^\text{m} \left( \rho x_1^\text{b} n_2^\text{s} + \delta x_1^\text{b} n_2^\text{exit} \right).\label{eq:Jb}
\end{align}
\fi

\noindent For simplicity of reference and analysis, we rewrite Equation~\eqref{eq:Js},~\eqref{eq:Jb},~\eqref{eq:J_2^s},~\eqref{eq:J_2^exit}, and~\eqref{eq:J_0^enter} as follows, where $K_i^*$ ($i \in [0,1,2]$, $* \in \{\text{s}, \text{b}, \text{exit}, \text{enter}\}$) aggregates all terms associated with variables $x_1^\text{s}$ and $x_1^\text{b}$, $B_i^*$ ($i \in [0,1,2]$, $* \in \{\text{s}, \text{b}, \text{exit}, \text{enter}\}$) collects the remaining constant terms:
\begin{align}
  J_{1}^\text{s} (\textbf{x}) &= K_{1}^\text{s} x_{1}^\text{s} + B_{1}^\text{s}, \label{eq:J_1^s_x_1^s}\\
  J_{1}^\text{b} (\textbf{x}) &= K_{1}^\text{b} x_{1}^\text{b} + B_{1}^\text{b}, \label{eq:J_1^b_x_1^b}\\
  J_{2}^\text{s} (\textbf{x}) &= K_{2}^\text{s} x_{1}^\text{b} + B_{2}^\text{s}, \\
  J_{2}^\text{exit} (\textbf{x}) &= K_{2}^\text{exit} x_{1}^\text{s} + B_{2}^\text{exit}, \\
  J_{0}^\text{enter} (\textbf{x}) &= K_{0}^\text{enter} x_{1}^\text{s} + B_{0}^\text{enter},
\end{align}

%As traffic volume on a lane increases, the behavioral cost of entering that lane rises due to intensified merging and traversing conflicts. Accordingly, all congestion-sensitivity coefficients $K_\text{1}^\text{s}, K_\text{1}^\text{b}, K_\text{1}^\text{exit}$, $ K_2^\text{s}, K_2^\text{exit}$ and $K_0^\text{enter}$ are non-negative, since they scale with the flow variables $x_1^\text{s}$ and $x_1^\text{b}$. In contrast, the constants $B_1^\text{s}, B_1^\text{b}, B_1^\text{exit}, B_2^\text{s}, B_2^\text{exit},$ and $B_0^\text{enter}$ are also non-negative and capture fixed costs determined solely by the geometric configuration of the weaving ramp.

To explicitly characterize the control mechanism of dedicated altruistic CAVs, we further introduce $q_\text{s}$ as the CAV steadfast proportion, with the remaining fraction $1-q_\text{s}$ corresponding to bypassing proportion. Here, $q_\text{s}$ is controlled by the central command. When $q_\text{s} = 1$, all CAVs remain steadfast; when $q_\text{s}=0$, all adopt bypassing. For intermediate values of $q_\text{s} \in (0,1)$, a mixed strategy arises in which some CAVs remains steadfast while others bypass. Accordingly, the CAV flow proportions are given by:
\begin{align}
  x_{1,\text{CAV}}^\text{s} &= p q_\text{s}, \label{eq:choice_CAV_s} \\
  x_{1,\text{CAV}}^\text{b} &= p (1-q_\text{s}), 
\end{align}

Each behavioral cost $J$ is a function of the CAV penetration rate $p$ and the CAV steadfast proportion $q_\text{s}$. This formulation implies that the system-level delay can be directly influenced through the strategic control of through CAVs on Lane 1. In particular, adjusting $q_\text{s}$ for a given penetration rate $p$ provides a lever for mitigating congestion and improving social efficiency.

\subsection{The Stackelberg-Wardrop Framework}

\noindent The interaction between CAVs and HDVs can be naturally represented by a \textbf{leader-follower structure}, reflecting their distinct behavioral characteristics. CAVs, with access to global system information through centralized controllers, can make informed decisions in advance, whereas HDVs rely solely on local observations of surrounding vehicles and respond based on observed behaviors. 

\textbf{To explicitly capture this leader–follower dynamic, a bilevel optimization framework is more appropriate than a single-level equilibrium model.} Specifically, we adopt a Stackelberg game formulation~\citep{basar1999dynamic}, in which the leader first commits to a strategy and the follower subsequently responds optimally given the leader's choice. This hierarchical structure aligns closely with the interaction logic between CAVs and HDVs in mixed-autonomy traffic systems. Within this framework, dedicated altruistic CAVs act as leaders, aiming to minimize the system's social cost (total delay), while selfish HDVs act as followers, individually seeking to minimize their own delays. The collective responses of HDVs are captured by the Wardrop equilibrium conditions previously introduced in Definition~\ref{def:wdp_basic}.

Thus, we propose a Stackelberg-Wardrop bilevel framework to model strategic lane choice behavior in mixed autonomy. In this framework, the upper-level is formulated as a Stackelberg game, where CAVs act as leaders and HDVs as followers. The lower-level follows the Wardrop equilibrium principle, capturing the adaptive responses of HDVs to the strategies selected by CAVs. The Stackelberg-Wardrop framework is defined as follows: 

\begin{definition}[Stackelberg--Wardrop Strategic Lane Choice Equilibrium]
\label{def:stack_opt}
For a given weaving ramp configuration $G=(\mathbf{N},\mathbf{C})$ and CAV penetration rate $p\in[0,1]$, 
let the Stackelberg--Wardrop flow vector be 
$\mathbf{x}:=\big(x^{\mathrm{s}}_{1,\mathrm{CAV}}=pq_\text{s},\,x^{\mathrm{b}}_{1,\mathrm{CAV}}=p(1-q_\text{s}),\,x^{\mathrm{s}}_{1,\mathrm{HDV}},\,x^{\mathrm{b}}_{1,\mathrm{HDV}}\big)$.
A flow vector $\mathbf{x}^{*}$ is a Stackelberg--Wardrop equilibrium if and only if the following bilevel
conditions hold:

\begin{itemize}
\item \textit{Upper level (CAV optimization).}
\begin{align}
q_s^* \in \arg\min_{q_s \in [0,1]} 
J_{\mathrm{soc}}\!\big(\mathbf{x}^*\big). \label{eq:upper_xb}
\end{align}

\item \textit{Lower level (HDV lane-choice equilibrium).}
\begin{subequations}\label{eq:wardrop_lower}
\begin{align}
x^{\mathrm{s}}_{1,\mathrm{HDV}}\!\left(J^{\mathrm{s}}_{1}(\mathbf{x^*})-J^{\mathrm{b}}_{1}(\mathbf{x^*})\right) &\le 0, \label{eq:equilibrium1}\\
x^{\mathrm{b}}_{1,\mathrm{HDV}}\!\left(J^{\mathrm{b}}_{1}(\mathbf{x^*})-J^{\mathrm{s}}_{1}(\mathbf{x^*})\right) &\le 0, \label{eq:equilibrium2}.
\end{align}
\end{subequations}
\end{itemize}
\end{definition}

%This definition formalizes the Stackelberg-Wardrop equilibrium as a bilevel structure that integrates centralized CAV control with decentralized HDV responses. The upper-level problem minimizes the social cost with respect to the CAV steadfast proportion $q_\text{s}$, while the objective also depends on HDV behavior, which is endogenously determined by the HDV lane choice equilibrium. This results in a Mathematical Program with Equilibrium Constraints (MPEC), where the lower-level is formulated as a Mixed Complementary Problem (MCP) ~\citep{luo1996mathematical,outrata2013nonsmooth}. The interaction between HDVs and CAVs within this framework can be interpreted as an iterative adjustment process converging to equilibrium. CAVs first determine a tentative allocation between steadfast and bypassing strategies. HDVs then respond to these decisions by adjusting their own lane choices, which alters the overall system state. In response, CAVs update their strategy, and this process continues until neither CAVs nor HDVs have an incentive to deviate. 

This definition formalizes the Stackelberg–Wardrop equilibrium as a bilevel structure integrating centralized CAV control with decentralized HDV responses. The upper-level problem minimizes the social cost with respect to the CAV steadfast proportion $q_\text{s}$, while the objective implicitly depends on HDV behavior determined by the lower-level lane choice equilibrium. This leads to a Mathematical Program with Equilibrium Constraints (MPEC), where the lower level is formulated as a Mixed Complementarity Problem (MCP)~\citep{luo1996mathematical,outrata2013nonsmooth}. The interaction between HDVs and CAVs can be viewed as an iterative process: CAVs choose a tentative strategy allocation, HDVs respond by adjusting their lane choices, and the system iterates until neither side has an incentive to deviate.

\begin{remark}[Residual Reformulation for Efficient Computation]
  To reduce the computational burden of repeatedly solving the lower-level equilibrium \eqref{eq:equilibrium1} and \eqref{eq:equilibrium2}, we reformulate the complementarity conditions into continuous residual functions of the nonlinear complementarity problem (NCP), which allows the Sequential Quadratic Programming (SQP) algorithm to update upper-level variables without explicitly resolving the KKT system at each iteration. Specifically, the lane-choice equilibrium is enforced by requiring $h_1(x_{1,\text{HDV}}^\text{s})=0$ and $h_2(x_{1,\text{HDV}}^\text{b})=0$, where $h_1(x_{1,\text{HDV}}^\text{s})=x_{1,\text{HDV}}^\text{s} \cdot \max(0,J_1^\text{s}-J_1^\text{b})$, and $h_2(x_{1,\text{HDV}}^\text{b})=x_{1,\text{HDV}}^\text{b} \cdot \max(0,J_1^\text{b}-J_1^\text{s})$. Hence, by converting the bilevel problem into a tractable single-level nonlinear program, the proposed residual reformulation enables efficient computation without repeatedly solving the lower-level equilibrium problem.
\end{remark}

\subsection{Impact of Dedicated Altruistic CAVs}\label{subsec:HDV-CAV}
%This phase analysis characterizes the interaction between CAVs and HDVs and highlights how the increasing penetration of CAVs gradually influences HDV behavior. When the CAV penetration rate is very low, CAVs have limited influence on the traffic system and are often compelled to bypass due to the dominant behavior patterns of HDVs. As the penetration rate increases, CAVs begin to exhibit a behavioral shift, transitioning from bypassing to steadfast strategies. Within a certain range of penetration rates, all CAVs adopt the steadfast behavior, contributing to a more stable flow configuration. However, when the penetration rate becomes sufficiently high, further changes in HDV behavior are triggered, potentially altering the system dynamics once again.

%To formalize the behavioral transitions, we introduce three critical penetration thresholds: \textbf{Influence threshold $p_1$}, the minimum CAV penetration rate at which CAVs begin to alter HDV lane-choice decisions. \textbf{Efficiency threshold $p_2$}, the penetration level at which the social delay $J_\text{soc}$ first starts to decrease. \textbf{Saturation threshold $p_3$}, the penetration level beyond which further CAV deployment no longer yields additional reductions in $J_{\text{soc}}$. In general, the set {$p_1,p_2,p_3$} partitions the penetration spectrum $p \in (0,1]$ into distinct interaction regimes, each reflecting the evolving role of CAVs in mixed-autonomy traffic.

\noindent For the phase analysis, we denote two steadfast proportion $x_1^\text{s}$ benchmarks as the \textbf{pure HDV equilibrium} at HDV-only environment ($p=0$), denoted by $\Phi$, and the \textbf{socially optimal equilibrium}, where total cost reaches the theoretical minimum, denoted by $\Gamma$, respectively. Our focus is on configurations satisfying $\Phi<\Gamma$, which implies that under purely selfish behavior, fewer vehicles choose the steadfast strategy than in the socially optimal case (echoed by results in Section~\ref{fig:validation}). 
%Furthermore, we define the steadfast proportion $x_1^\text{s}$ benchmark at \textbf{pure CAV equilibrium} when $p=1$ as $\Psi$, we have $\Psi=\Gamma$ in this section, indicating that dedicated altruistic CAVs are fully altruism. 
The admissible set of configurations for the subsequent analysis is therefore refined as follows:
\begin{align}\label{eq:config}
  \mathcal{G}=\{G:0<\Phi<\Gamma<1\},
\end{align}

To formalize behavioral transitions, we introduce two critical penetration thresholds: \textbf{Efficiency threshold $p_1$}, defined as the minimum penetration rate at which the social delay $J_\text{soc}$ begins to decline. \textbf{Saturation threshold $p_2$}, defined as the penetration rate beyond which additional CAV deployment no longer produces further reductions in $J_{\text{soc}}$. For reference, we define the social delay in the absence of CAVs (i.e., $p=0$) as the baseline social delay $J_\text{soc}^\text{ref}$.

\begin{theorem}[Impact of Dedicated Altruistic CAVs on Social Delay] \label{thm:CAV_penetration}
For any given weaving ramp configuration \( G = (\mathbf{N}, \mathbf{C}) \in \mathcal{G} \), with CAV penetration rate $p$,
\begin{itemize}
  \item The social delay remains invariant under dedicated altruistic CAVs, i.e., $J_\text{soc}(\mathbf{x})=J_\text{soc}^\text{ref}$, if and only if $p \in \mathcal{A}_1$, where $\mathcal{A}_1:=[0,p_1]$,
  \item The social delay is decreased by dedicated altruistic CAVs, i.e., $J_\text{soc}(\mathbf{x})<J_\text{soc}^\text{ref}$, if and only if $p \in \mathcal{A}_2$, where $\mathcal{A}_2:=(p_1,p_2]$,
  \item The social delay is optimized by dedicated altruistic CAVs, i.e., $J_\text{soc}(\mathbf{x})=J_\text{soc}^\text{opt}$, if and only if $p \in \mathcal{A}_3$, where $\mathcal{A}_3:=(p_2,1]$,
\end{itemize}
where the two CAV penetration thresholds can be calculated as:
\begin{align}
    p_1&= \frac{K_1^\text{b}+B_1^\text{b}-B_1^\text{s}}{K_1^\text{s}+K_1^\text{b}},\\
    p_2&=\frac{2K_1^\text{b} + B_1^\text{b} - B_1^\text{s} -n_2^\text{exit}K_2^\text{exit} - n_0^\text{enter}K_0^\text{enter}+ n_2^\text{s} K_2^\text{s}}{2(K_1^\text{s}+K_1^\text{b})},
    %p_2&=\frac{2K_2-K_3+K_4}{2(K_1+K_2)}.
\end{align}
\end{theorem}

\begin{proof} 
The detailed proof is provided in Appendix~\ref{sec:appendix_theorem2}. The proof examines all possible cases of mixed autonomy equilibrium as the CAV penetration rate $p$ varies and analyzes the corresponding behaviors of both HDVs and CAVs. It shows that the social delay $J_\text{soc}$ remains unchanged for $p\in[0,p_1]$, decreases for $p\in(p_1,p_2]$, and reaches its minimum for $p\in(p_2,1]$. Furthermore, the theorem demonstrates that introducing dedicated altruistic CAVs effectively reduces the system delay and provides actionable guidance on the deployment scale (indicated by $p$) and locations (indicated by parameters $K$ and $B$) required to ensure system-level benefits.
\end{proof}

\begin{remark}[Flat Stage of Social Delay] \label{rem:flat_stage}
  When CAV penetration is low, the system experiences an initial plateau in the social cost \( J_{\text{soc}} \). This occurs because selfish HDVs take advantage of the limited number of altruistic CAVs, benefiting individually while preventing immediate system-wide improvement.
\end{remark}

\subsection{Numerical Example for Stackelberg-Wardrop Framework}

\noindent In this subsection, we present numerical examples to illustrate the impact of dedicated altruistic CAVs on both individual behavior and system-level performance as the CAV penetration rate $p$ varies. The traffic demand is set to 60 entering vehicles on Lane 0, 30 exiting vehicles on Lane 1, 30 exiting vehicles on Lane 2, and 200 through vehicles on Lane 1 and Lane 2 per hour. 

\begin{figure*}[h!]
\centering
\begin{subfigure}{0.31\textwidth}
    \includegraphics[width=\textwidth]{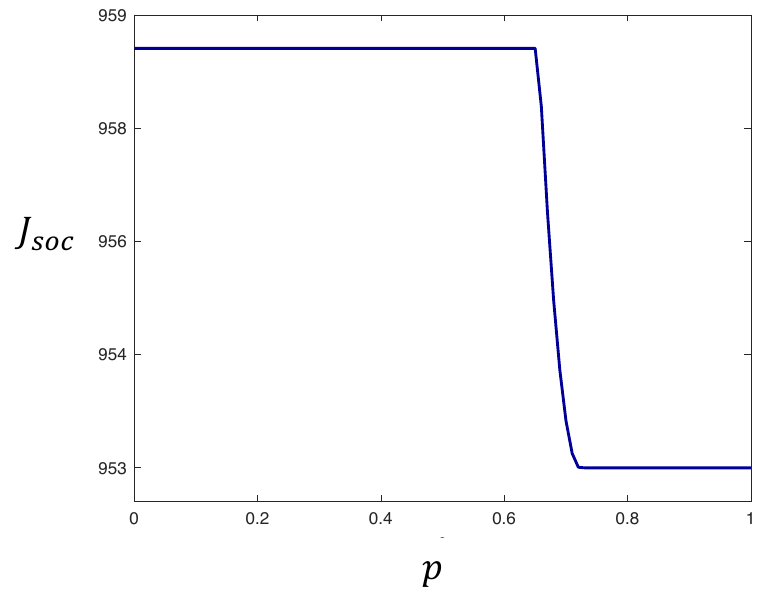}
    \caption{Social cost under $p$}
    \label{fig:alpha_social_cost}
\end{subfigure}
\hspace{0.01\textwidth} 
\begin{subfigure}{0.3\textwidth}
    \includegraphics[width=\textwidth]{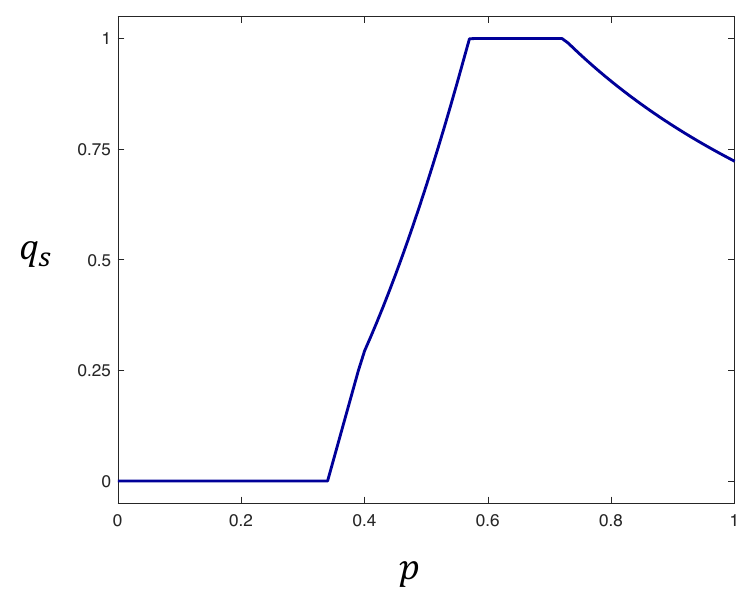}
    \caption{CAV strategy under $p$}
    \label{fig:alpha_beta}
\end{subfigure}
\hspace{0.01\textwidth} 
\begin{subfigure}{0.31\textwidth}
    \includegraphics[width=\textwidth]{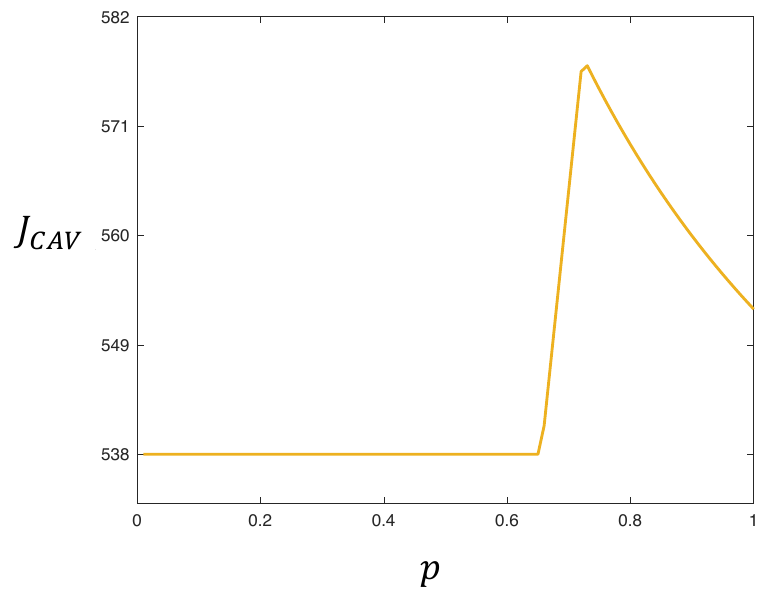}
    \caption{Individual cost under $p$}
    \label{fig:CAV_individual}
\end{subfigure}
\caption{These figures show the system-level and individual-level impact of CAV penetration rate $p$. Figure (a) shows how social cost $J_\text{soc}$ changes during $p$ increases. Figure (b) demonstrates how CAV steadfast proportion $q_\text{s}$ changes when $p$ increases. Figure (c) illustrates how each CAV's cost changes when $p$ increases.} 
\end{figure*}

Figure~\ref{fig:alpha_social_cost} illustrates how the social cost $J_\text{soc}$ varies with the CAV penetration rate $p$. When $p$ is low, $J_{soc}$ remains constant. It starts to decrease around $p \approx 0.6$, reaches its optimal value near $p \approx 0.7$, and remains at this optimal level as $p$ continues to increase. This pattern is consistent with the theoretical results established in Theorem~\ref{thm:CAV_penetration}.

Furthermore, we examine both $q_\text{s}$ in Figure \ref{fig:alpha_beta} and $J_\text{CAV}$ in Figure \ref{fig:CAV_individual}, where $J_\text{CAV}=p(J_1^\text{s}x_1^\text{s}+J_1^\text{b}x_1^\text{b})$) defined as individual CAV cost. When $p$ is low, $q_\text{s}$ remains constant initially and then begins to adjust. During this stage ($0 \leq p \leq 0.6$), $J_\text{soc}$ and $J_\text{CAV}$ remain unchanged. This indicates that CAVs adjust their strategies to benefit the system while the overall social cost stays constant, mainly because HDVs exploit the altruistic contributions of CAVs, which aligns with Remark~\ref{rem:flat_stage}. As $p$ further increases, such reduction comes at the cost of a significant increase in their individual cost, demonstrating their altruism toward society at their own expense. Finally, once $J_\text{soc}$ reaches its minimum, $J_\text{CAV}$ begins to decline, suggesting that dedicated altruistic CAVs can eventually benefit themselves once their market share becomes sufficiently large.

\section{AVs as Relaxed Altruistic Decision Makers}\label{sec:svo}

\noindent In practice, CAVs and HDVs rarely adhere to strictly homogeneous behavioral rules. Instead, they exhibit heterogeneous characteristics in decision-making and interactions. A more realistic framework therefore considers relaxed altruistic CAVs, whose objectives reflect a weighted combination of self-interest and social welfare. For HDVs, we characterize their behavioral diversity using the Social Value Orientation (SVO) framework~\citep{van1999pursuit}, which encompasses altruistic, selfish, competitive, and other extreme orientations. In this section, we analyze the performance of relaxed altruistic CAVs under varying penetration rates $p$ within heterogeneous environments, compare their outcomes with those of dedicated altruistic CAVs, and discuss potential mechanisms to incentivize altruistic behavior to enhance overall system efficiency and CAV–HDV interactions.

\subsection{Updated Notations for Relaxed Altruistic CAVs}

\begin{figure*}[h!]
\centering
\includegraphics[width = 0.5\textwidth]{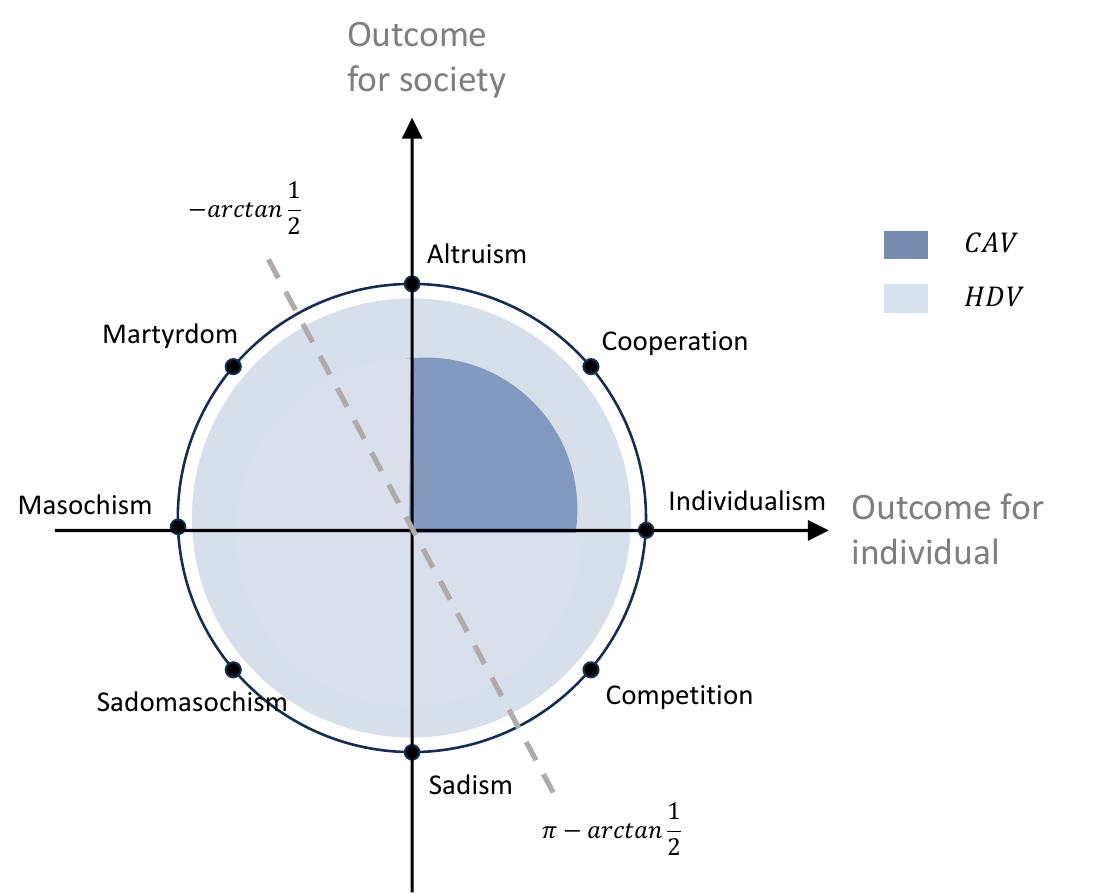}
\caption{This figure illustrates the SVO circle~\citep{liebrand1986value}, where each vehicle type $j$ is characterized by an angle $\theta_j$ representing its orientation toward individual versus collective outcomes. The x-axis denotes concern for individual outcomes, and the y-axis denotes concern for societal outcomes. Positive values indicate prosocial orientations, negative values reflect competitive or antagonistic tendencies, and zero represents neutrality.
In our setting, we focus on the subset of orientations satisfying $\cos\theta_j + \sin\theta_j > 0$, which corresponds to vehicle types whose combined valuation of individual and social outcomes is positive.}
\label{fig:SVO_circle}
\end{figure*}

\noindent As we extend the analysis from dedicated altruistic CAVs and selfish HDVs to a more heterogeneous setting, both CAVs and HDVs are represented by vehicles types with different SVO levels. Accordingly, we denote $j$ as the vehicle type, and update the steadfast and bypassing flow ratios as $x_{1,j}^\text{s}$ and $x_{1,j}^\text{b}$, respectively, where each type $j$ is associated with a specific SVO level $\theta_j$. 

Furthermore, we define the cost functions of SVO-level vehicles as a weighted combination of individual travel delay and marginal social cost, capturing their trade-off between self-interest and social welfare. The marginal social cost is defined as the additional total system delay generated by one more vehicle adopting a given strategy~\citep{beckmann1956studies}, which we mathematically represent as the partial derivative of the social cost with respect to the corresponding flow. Let $\tilde{J}_{1,j}^\text{s}$ and $\tilde{J}_{1,j}^\text{b}$ denote the cost function for steadfast and bypassing vehicles, respectively. To capture the heterogeneity in social preferences, we introduce the SVO parameter $\theta_j$, illustrated in Figure \ref{fig:SVO_circle}. Specifically, $\theta_j=0$ corresponds to purely selfish behavior, $\theta_j=\frac{\pi}{2}$ represents fully altruistic behavior, and $\theta_j=(0,\frac{\pi}{2})$ indicates partially altruistic behaviors. Extreme altruism is described by $\theta_j \in (\frac{\pi}{2},\pi]$. Negative values capture non-cooperative orientations: $\theta_j \in (-\frac{\pi}{2},0)$ denotes competitive behavior, and $\theta_j \in (-\pi,-\frac{\pi}{2})$ corresponds to extreme sadomasochism tendencies, which may be interpreted as pathological driving behaviors. In this formulation, $sin(\theta_j)$ captures the vehicle's orientation toward social welfare, while $cos(\theta_j)$ reflects the weight placed on individual outcomes. Thus, we have the cost function for steadfast and bypassing vehicles $\tilde{J}_{1,j}^\text{s}$ and $\tilde{J}_{1,j}^\text{b}$ as:
\begin{align} 
   \tilde{J}_{1,j}^\text{s}(\mathbf{x}) &= cos(\theta_j)J_{1}^\text{s}(\mathbf{x})+ sin(\theta_j)\frac{\partial J_\text{soc}(\mathbf{x})}{\partial x_{1,j}^\text{s}}, \label{eq:SVO_cost1} \\
   \tilde{J}_{1,j}^\text{b}(\mathbf{x}) &= cos(\theta_j)J_1^\text{b}(\mathbf{x}) + sin(\theta_j)\frac{\partial J_\text{soc}(\mathbf{x})}{\partial x_{1,j}^\text{b}}, \label{eq:SVO_cost2} 
\end{align}

HDVs and CAVs are distinguished by their respective ranges of the SVO parameter~$\theta_j$. For relaxed altruistic CAVs, $\theta_j$ is constrained within $0 \leq \theta_j \leq \frac{\pi}{2}$, reflecting their controllable and cooperative nature. In contrast, HDVs are assigned $\theta_j$ satisfying $cos \theta_j+2sin\theta_j>0$, capturing the broader behavioral heterogeneity of human drivers.

\iffalse
Based on $\tilde{J}_{1,j}^\text{s}$ and $\tilde{J}_{1,j}^\text{b}$, we can further update the social cost $J_\text{soc}$ from Equation~\eqref{eq:social_cost} as:
\begin{align}
  \tilde{J}_\text{soc}=\sum_{j=1}^J x_{1,j}^\text{s} \tilde{J}_{1,j}^\text{s}+ \sum_{j=1}^J x_{1,j}^\text{b} \tilde{J}_{1,j}^\text{b} + n_2^\text{s}J_2^\text{s}+n_2^\text{exit}J_2^\text{exit}+n_0^\text{enter}J_0^\text{enter}.\label{eq:hetero_social_cost}
\end{align}
\fi

\subsection{Heterogeneous Coupled Wardrop Framework}

\noindent In the heterogeneous setting, both CAVs and HDVs adopt cost functions defined in Equations~\eqref{eq:SVO_cost1} and~\eqref{eq:SVO_cost2}, respectively. All vehicle types are assumed to satisfy Wardrop’s principle, leading to the following framework:
\\
\begin{definition} [Heterogeneous Coupled Wardrop Lane Choice Equilibrium]\label{def:wdp_SVO}
  For a given weaving-ramp configuration \( G = (\mathbf{N}, \mathbf{C}) \), a flow distribution vector \( \mathbf{x} \) is in equilibrium if and only if
\begin{subequations}\label{eq:eq_def}
    \begin{align}
        x_{1,j}^\text{s}  (\tilde{J}_{1,j}^\text{s}(\mathbf{x}) - \tilde{J}_{1,j}^\text{b}(\mathbf{x})) &\leq 0 ,\\
        x_{1,j}^\text{b}  (\tilde{J}_{1,j}^\text{b}(\mathbf{x}) - \tilde{J}_{1,j}^\text{s}(\mathbf{x})) &\leq 0 ,
    \end{align}
\end{subequations}
for all vehicle classes $j=1,2,\ldots,J$.
\end{definition}

\subsection{Impact of Relaxed Altruistic CAVs}

%For the phase analysis, we denote total steadfast proportion $x_1^\text{s}$ at $p=0$ as the pure HDV equilibrium $\Phi$, at $p=1$ as the pure CAV equilibrium, and the socially optimal equilibrium as $\Gamma$. In the heterogeneous scenario, $\Psi \leq \Gamma$. Thus, the admissible set of configurations for the subsequent analysis is:
%\begin{align}
%  \mathcal{G}=\{G:0<\Phi<\Psi\leq \Gamma<1\},
%\end{align}

\noindent There are $H$ types of HDVs and $C$ types of CAVs, indexed by $h \in \mathcal{H}$ 
and $c \in \mathcal{C}$ respectively, where $\mathcal{H} \cup \mathcal{C} = \mathcal{J}$ 
and $|\mathcal{J}| = J$. Let $w_h^\text{HDV}$ denote the proportion of type $h$ within 
the HDV subpopulation, and $w_c^\text{CAV}$ denote the proportion of type $c$ within 
the CAV subpopulation, satisfying $\sum_{h \in \mathcal{H}} w_h^\text{HDV} = 1$ and 
$\sum_{c \in \mathcal{C}} w_c^\text{CAV} = 1$. In the following analysis, we consider $w_h^\text{HDV}$ and $w_c^\text{CAV}$ as given and fixed. Given a CAV penetration rate $p \in [0,1]$, 
the overall population share of each vehicle type $j \in \mathcal{J}$ is defined as:
\begin{align}
    w_j(p) = 
    \begin{cases} 
        (1-p)\, w_h^\text{HDV} & \text{if } j = h \in \mathcal{H}, \\[4pt]
        p\, w_c^\text{CAV}     & \text{if } j = c \in \mathcal{C},
    \end{cases}
    \label{eq:population_share}
\end{align}
where $w_j(p)$ represents the fraction of type $j$ vehicles in the overall traffic 
population at penetration rate $p$. It follows directly that 
$\sum_{j \in \mathcal{J}} w_j(p) = 1$ for all $p \in [0,1]$.

\begin{theorem}[Impact of Relaxed Altruistic CAVs on Social Delay]
\label{thm:theorem3}
Consider a weaving-ramp configuration $G=(N,\mathcal{C})$ with CAV penetration 
rate $p\in[0,1]$ and vehicle types $j=1,\cdots,J$. Assume that for every vehicle 
type $j$, $\cos(\theta_j)+2\sin(\theta_j)>0$, and that the thresholds 
$\{\chi_j\}_{j=1}^{J}$, defined by
\begin{align}
    \chi_j := \frac{B_{1,j}^\text{b}+K_{1,j}^\text{b}-B_{1,j}^\text{s}}
    {K_{1,j}^\text{s}+K_{1,j}^\text{b}}, \quad j=1,\cdots,J,
\end{align}
are pairwise distinct. For each admissible index $k$, define the aggregate 
steadfast share of all types with threshold exceeding $\chi_k$ and the 
associated plateau interval as:
\begin{align}\label{eq:I_k}
    W_k(p) := \sum_{j:\,\chi_j>\chi_k} w_j(p), 
    \qquad 
    I_k := \bigl\{p\in[0,1] : 0 < \chi_k - W_k(p) < w_k(p)\bigr\},
\end{align}
where $w_j(p)$ is the overall population share of type $j$ as defined in 
\eqref{eq:population_share}. Then the following statements hold:
\begin{itemize}
    \item For any $p\in[0,1]$, at most one vehicle type can play mixed strategy at equilibrium.

    \item Type $k$ vehicles are the unique type that plays the mixed strategy at equilibrium, if and only if conditions~\eqref{eq:I_k} are satisfied. At the equilibrium, all 
    types with $\chi_j>\chi_k$ are fully steadfast, and all types with 
    $\chi_j<\chi_k$ are fully bypass.

    \item On $I_k$, the equilibrium steadfast proportion satisfies 
    ${x_1^\text{s}}^*=\chi_k$. Since $J_\text{soc}$ depends only on 
    $x_1^\text{s}$, it follows 
    that $J_\text{soc}$ remains constant throughout $I_k$.

\iffalse
    \item The boundaries of the plateau interval $I_k$ are determined by 
    the two linear equations:
    \begin{align}
        W_k(p) = \chi_k, \qquad W_k(p) + w_k(p) = \chi_k.
        \label{eq:Ik_boundaries}
    \end{align}
    The onset and termination of the plateau associated with type $k$ 
    correspond to the penetration rates at which type $k$ first becomes 
    active and subsequently transitions to a purely steadfast or purely 
    bypass strategy.
\fi

    \item A penetration range $\mathcal{P}$ is free of 
    constant social-delay intervals if and only if 
    $\mathcal{P}\cap I_k=\varnothing$ for all $k=1,\cdots,J$.
\end{itemize}
\end{theorem}

\begin{proof}
    The detailed proof is provided in Appendix~\ref{sec:appendix_theorem3}. As a sketch proof, we describe the key points as follows. First, vehicle types are ordered by their thresholds $\chi_j$ in ascending order, which determines a priority ranking that governs which type becomes the unique mixed type at any given penetration rate $p$. Second, for each interval $I_k$, the proof characterizes the equilibrium structure by showing that all types with $\chi_j > \chi_k$ are fully steadfast, all types with $\chi_j < \chi_k$ are fully bypass, and only type $k$ adopts a mixed strategy. Third, since the aggregate steadfast proportion is pinned at ${x_1^\text{s}}^* = \chi_k$ throughout $I_k$, and $J_\text{soc}$ depends only on $x_1^\text{s}$, it follows that $J_\text{soc}$ remains constant on each such interval. Consequently, improvements in social cost occur only at the boundary points of $I_k$, where the active mixed type switches from type $k$ to an adjacent type.
\end{proof}

The above results provide a structural characterization of penetration ranges that induce constant social delay. In particular, each interval $I_k$ corresponds to a plateau governed by the behavioral threshold $\chi_k$ of vehicle type $k$, and the union $\mathcal{I} := \bigcup_{k=1}^{J} I_k$ captures the full collection of such undesirable regimes. This characterization directly enables the design of target penetration ranges $\mathcal{P}$ that avoid all constant-delay regions, i.e., $\mathcal{P} \cap \mathcal{I} = \varnothing$, thereby ensuring that social cost improvements are achievable as $p$ varies within $\mathcal{P}$.

These findings reveal a fundamental limitation of penetration-rate-only deployment strategies: under heterogeneous behavioral responses, increasing $p$ alone does not guarantee monotonic improvements in system performance. Instead, $J_\text{soc}$ exhibits a piecewise-constant structure, remaining unchanged within each interval $I_k$ and decreasing only at the boundary points where the active mixed type switches. The emergence of these plateaus is intrinsically linked to the distribution of behavioral preferences $\theta_j$ across vehicle types, as each plateau $I_k$ is uniquely associated with the threshold $\chi_k$ of its governing type. Consequently, effective CAV deployment must account not only for CAV penetration rate $p$, but also for the composition of behavioral types within the CAV fleet, so as to position the operating point outside $\mathcal{I}$ and realize meaningful reductions in social cost.

\subsection{Numerical Example for Heterogeneous Wardrop Framework}

\noindent In this subsection, we do numerical examples to show comparisons impacts between relaxed altruistic CAVs and dedicated altruistic CAVs, and demonstrate insights in Theorem \ref{thm:theorem3}. In this simulation, we conduct 30 entering vehicles on Lane 0, 10 exiting vehicles on Lane 1, 10 exiting vehicles on Lane 2, and 50 through vehicles on Lane 1 and Lane 2 per hour. 

\begin{figure*}[h!]
\centering
\begin{subfigure}{0.45\textwidth}
    \includegraphics[width=\textwidth]{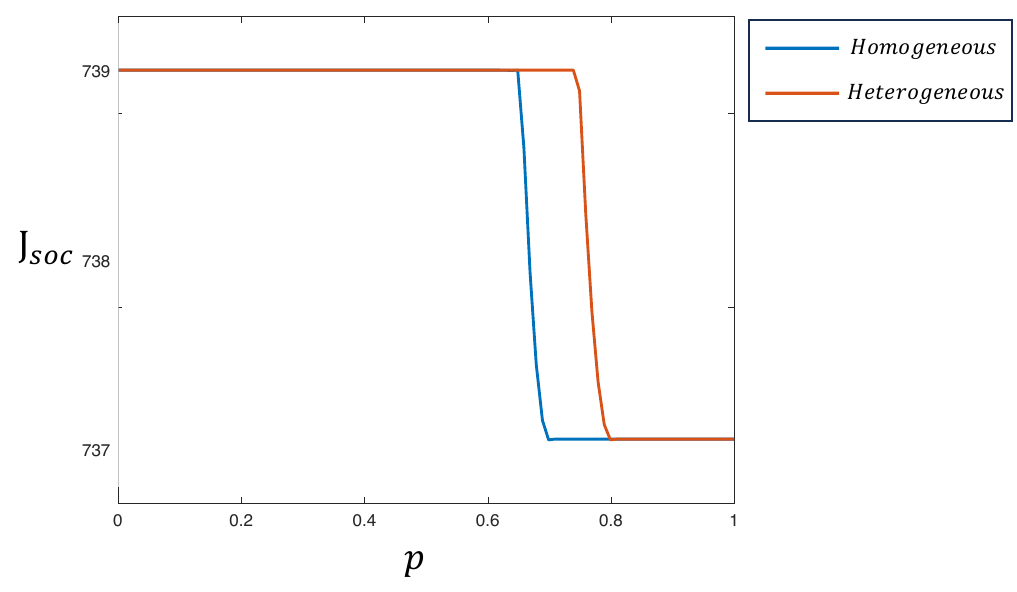}
    \caption{Heterogeneous vs. Homogeneous Phase}
    \label{fig:compare_phase}
\end{subfigure}
\hspace{0.01\textwidth} 
\begin{subfigure}{0.41\textwidth}
    \includegraphics[width=\textwidth]{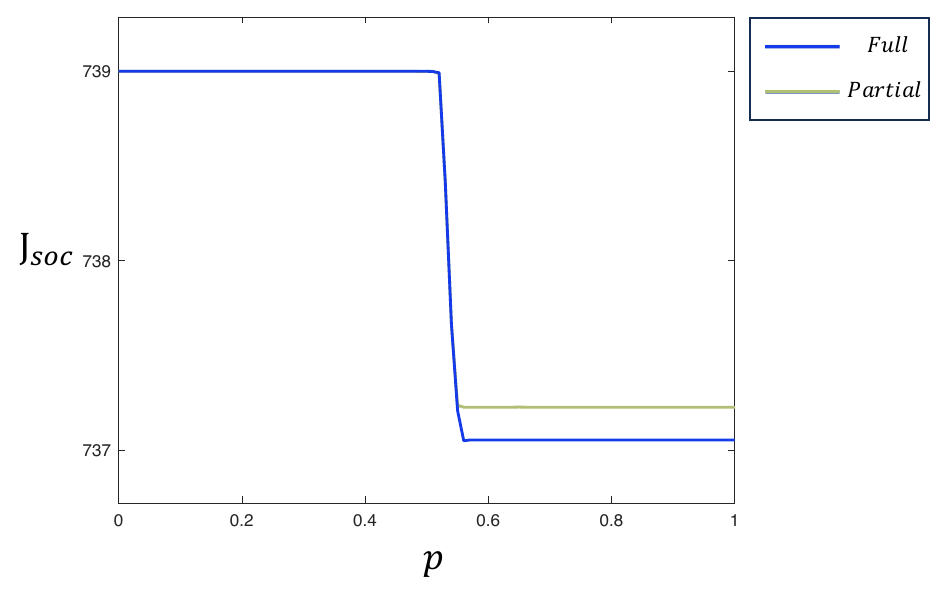}
    \caption{Full vs. Partial Altruism Phase}
    \label{fig:compare1}
\end{subfigure}
\caption{Figure (a) presents the social cost phase comparison between homogeneous and heterogeneous scenarios. The homogeneous scenario involves single type of dedicated altruism CAVs, whereas the heterogeneous scenario consists of multiple types of relaxed altruistic CAVs. Figure (b) provides a comparison between relaxed altruistic CAVs and dedicated altruistic CAVs in the homogeneous scenario. }\label{fig:compare}
\end{figure*}

For Figure~\ref{fig:compare_phase}, we compare two settings: (i) a homogeneous scenario with a single type of relaxed altruistic CAVs and a single type of selfish HDVs, and (ii) a heterogeneous scenario involving multiple types of relaxed altruistic CAVs and HDVs. In both cases, the social cost initially remains stable, then decreases, and ultimately converges toward the optimal level, consistent with the trend described in Theorem~\ref{thm:theorem3}. Notably, in the heterogeneous scenario, the onset of social cost reduction is delayed, beginning near $p=0.75$, and stabilizing around $p=0.8$, compared to the homogeneous case, where the transition occurs earlier, between $p=0.65$ and $p=0.7$. This delay highlights how the diversity of vehicle types and behavioral heterogeneity can influence the rate and extent of system improvement as CAV penetration increases.

For Figure~\ref{fig:compare1}, we further examine the difference between dedicated altruistic CAVs and relaxed altruistic CAVs in homogeneous settings. While relaxed altruistic CAVs contribute to reducing the social cost $J_\text{soc}$, their impact plateaus at high penetration levels, leading to a noticeable gap relative to dedicated altruistic CAVs. This finding suggests that partial altruism, though beneficial, cannot fully replicate the societal gains achieved under complete altruism. From a policy perspective, this implies that moderate levels of altruistic design or incentive mechanisms may improve system performance without requiring fully altruistic vehicle control, but achieving the social optimum still necessitates stronger coordination or regulatory intervention. 
\begin{figure*}[h!]
\centering
\begin{subfigure}{0.34\textwidth}
    \includegraphics[width=\textwidth]{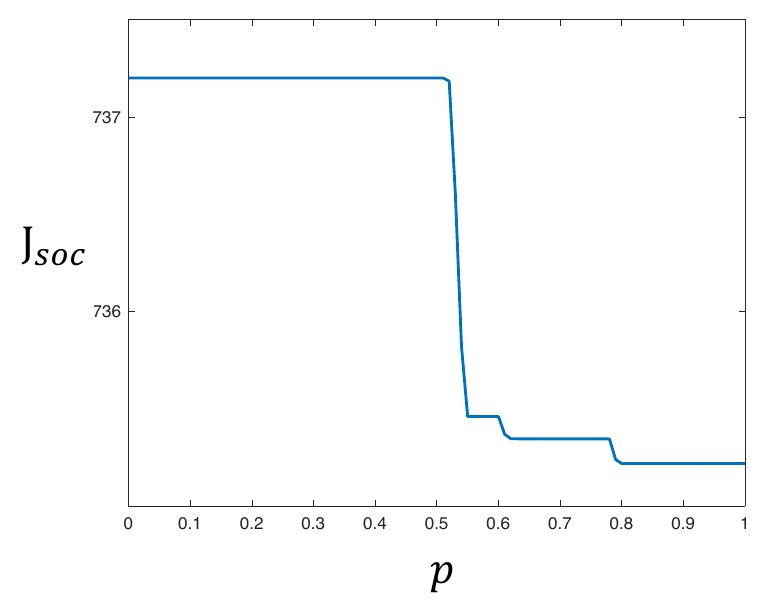}
    \caption{Social cost under $p$}
    \label{fig:hetero_social}
\end{subfigure}
\hspace{0.1\textwidth} 
\begin{subfigure}{0.345\textwidth}
    \includegraphics[width=\textwidth]{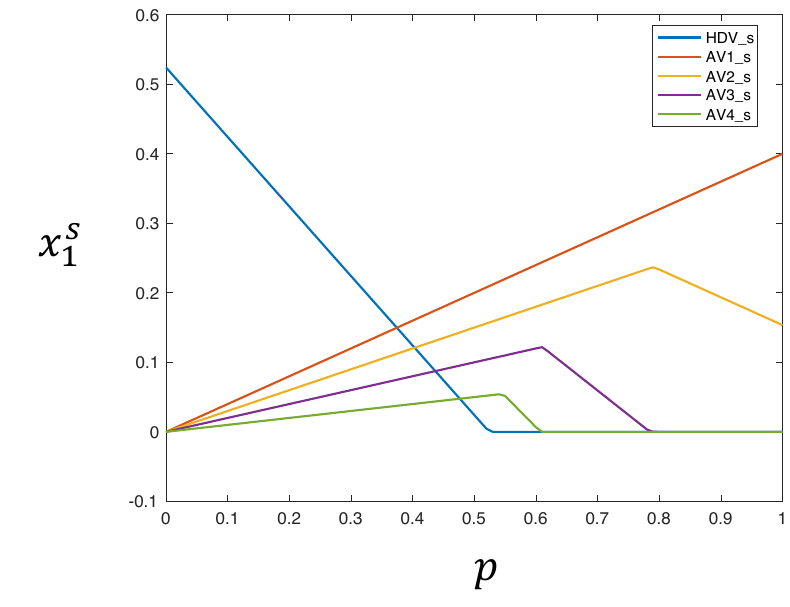}
    \caption{Vehicle strategies under $p$}
    \label{fig:hetero_strategy}
\end{subfigure}
\caption{Figure (a) presents the change of the social cost $J_{\text{soc}}$ as the CAV penetration rate $p$ varies in the heterogeneous mixed-autonomy setting. Figure (b) shows the corresponding change of each vehicle type $j$'s steadfast proportion $x_1^\text{s}$. The results illustrate that the multi-stage evolution of $J_{\text{soc}}$ corresponds to changes in the active mixed vehicle type $k$, which partitions the parameter range into several equilibrium regimes.}\label{fig:hetero}
\end{figure*}

%Furthermore, we design CAV types to be $\theta_{1,1}=\frac{1}{5}\pi,\theta_{1,2}=\frac{1}{4}\pi,\theta_{1,3}=\frac{1}{3}\pi,\theta_{1,4}=\frac{1}{2}\pi$, with within-CAV shares $w_1^\text{HDV}=0.1,w_2^\text{HDV}=0.2,w_3^\text{HDV}=0.3,w_4^\text{HDV}=0.4$. The change of social cost $J_\text{soc}$ and each vehicle type's steadfast proportion are shown in Figure~\ref{fig:hetero}. The results show that the multi-stage evolution of $J_\text{soc}$ corresponds to changes in the active mixed vehicle type. When the mixed vehicle type switches, the equilibrium steadfast proportion $x_1^\text{s}$ shifts to a new level, which leads to a corresponding change in the social cost. This numerical pattern is consistent with Theorem~\ref{thm:theorem3}, which shows that each interior interval $I_k$ is associated with a unique active mixed type and a constant value of $J_\text{soc}$, while regime switching occurs at the boundary penetration rates.

Furthermore, we design four types of CAVs with $\theta_{1,1}=\frac{1}{5}\pi$, $\theta_{1,2}=\frac{1}{4}\pi$, $\theta_{1,3}=\frac{1}{3}\pi$, and $\theta_{1,4}=\frac{1}{2}\pi$, with corresponding within-CAV shares $w_1^\text{CAV}=0.1$, $w_2^\text{CAV}=0.2$, $w_3^\text{CAV}=0.3$, and $w_4^\text{CAV}=0.4$. The evolution of the social cost $J_\text{soc}$ and the steadfast proportions of each vehicle type are illustrated in Figure~\ref{fig:hetero}. The results reveal a clear multi-stage structure in the evolution of $J_\text{soc}$ as the penetration rate $p$ increases. In particular, $J_\text{soc}$ remains constant within each interval and decreases only at specific transition points. This piecewise-constant behavior directly corresponds to the interval structure characterized in Theorem~\ref{thm:theorem3}, where each interior interval $I_k$ is associated with a unique active mixed vehicle type. 

From Figure~\ref{fig:hetero_strategy}, we observe that as $p$ increases, different vehicle types sequentially become the active mixed type in the equilibrium. Within each interval, only one vehicle type exhibits a mixed strategy, while all other types are either fully steadfast or fully bypass. As predicted by Theorem~\ref{thm:theorem3}, this implies that the aggregate steadfast proportion $x_1^\text{s}$ remains fixed at the corresponding threshold $\chi_k$, leading to a constant value of $J_\text{soc}$ in Figure~\ref{fig:hetero}(a). At the boundary of each interval, a regime-switching event occurs: the previously active mixed type becomes pure, and a new vehicle type takes over as the active mixed type. This transition causes a discrete shift in the equilibrium steadfast proportion $x_1^\text{s}$, which in turn produces a stepwise change in the social cost. As a result, the overall evolution of $J_\text{soc}$ follows a staircase-like pattern, where each plateau corresponds to a specific behavioral type. These results highlight that the multi-stage evolution of system performance is driven by the sequential activation of heterogeneous vehicle types. Rather than improving continuously with $p$, the system exhibits periods of stagnation (plateaus) followed by abrupt improvements at regime-switching points. This mechanism underscores the critical role of behavioral heterogeneity in shaping both the structure and timing of efficiency gains in mixed-autonomy traffic systems.

Figures~\ref{fig:compare} and~\ref{fig:hetero} jointly illustrate how behavioral heterogeneity shapes the evolution of the social cost as the CAV penetration rate increases. While both homogeneous and heterogeneous settings eventually achieve reductions in social cost, the heterogeneous case exhibits a delayed onset of improvement and a distinct multi-stage transition pattern. These results are consistent with Theorem~\ref{thm:theorem3}, where the penetration range is partitioned into intervals $I_k$ associated with different active vehicle types. As the penetration rate increases, regime switching between these intervals leads to stepwise changes in the equilibrium and the social cost. Overall, the results demonstrate that heterogeneous behavioral preferences not only affect the magnitude of efficiency gains but also induce a multi-stage, staircase-like evolution of system performance.

\section{Conclusions} \label{sec:future}

%This study develops a unified modeling framework to investigate how altruistic and strategic behaviors of CAVs affect lane choice and congestion at highway weaving ramps. Starting from a Wardrop formulation for HDVs, we construct a bilevel Stackelberg–Wardrop model capturing strategic CAV leadership and HDV equilibrium responses, and further extend it to a heterogeneous Wardrop framework incorporating SVO-based behavioral diversity. Analytical and numerical results consistently reveal that altruistic CAVs can effectively reduce social delay through a three-phase evolution, from HDV-dominated inefficiency to socially optimal operation, while behavioral heterogeneity moderates but does not negate the benefits of altruism.

%Future extensions will incorporate dynamic and stochastic environments to model time-varying demand, adaptive learning, and behavioral uncertainty, enabling analysis of temporal effects in realistic traffic systems. We also plan to generalize the framework to network-level applications such as multi-ramp corridors and interconnected mixed-autonomy networks. Ultimately, this research provides a foundation for designing incentive and regulatory mechanisms that harness strategic altruism to achieve equitable and efficient future mobility systems.

\noindent This study develops a unified modeling framework to investigate how altruistic and strategic behaviors of CAVs affect lane choice and congestion at highway weaving ramps. Starting from a Wardrop formulation for HDVs, we construct a bilevel Stackelberg--Wardrop model capturing strategic CAV leadership and HDV equilibrium responses, and further extend it to a heterogeneous Wardrop framework incorporating SVO-based behavioral diversity. 

Analytical and numerical results consistently reveal that altruistic CAVs can effectively reduce social delay, but the improvement is not monotonic. Instead, the system exhibits a multi-stage, piecewise-constant evolution, where the social cost remains unchanged within certain penetration intervals and decreases only at regime-switching points, which shows that each plateau corresponds to a specific active vehicle type. Behavioral heterogeneity reshapes these intervals, delaying the onset of improvements and altering the transition path, but does not negate the overall benefits of altruistic behavior.

Future extensions will incorporate dynamic and stochastic environments to model time-varying demand, adaptive learning, and behavioral uncertainty, enabling analysis of temporal effects in realistic traffic systems. We also plan to generalize the framework to network-level applications such as multi-ramp corridors and interconnected mixed-autonomy networks. Ultimately, this research provides a foundation for designing incentive and regulatory mechanisms that harness strategic altruism to achieve equitable and efficient future mobility systems.

%%%%%%%%%%%%%%%%%%%%%%%%%%%%%%%%%%%%%%%%%%%%%%%%%%%%%%%%%%%%%%%%%%%%%%%%%%%%%%%%
\iffalse
\section*{Acknowledgments}
This work was supported by 
The authors thank Matthew Wright for his helpful input during the preparation of this article.
\fi

%%%%%%%%%%%%%%%%%%%%%%%%%%%%%%%%%%%%%%%%%%%%%%%%%%%%%%%%%%%%%%%%%%%%%%%%%%%%%%%%

%\bibliographystyle{IEEEtran}
%\bibliographystyle{plainnat}
\section{Appendix} \label{sec:appendix}

\subsection{Proof for Theorem 1} \label{sec:appendix_hdv_eq}
\begin{proof}
To prove Theorem \ref{thm:uniq}, we first write the cost functions for through vehicles as follows:
\begin{align}
J_1^\text{s}(\mathbf{x}) &= C_1^\text{t} \left( \alpha x_1^\text{s} + \beta n_\text{exit} + n_\text{enter} \right)+ C_1^\text{m} \left(\omega  x_1^\text{s} n_2^\text{exit} +  x_1^\text{s} n_0^\text{enter} \right), \\
J_1^\text{b}(\mathbf{x}) &= C_2^\text{t} \left( \gamma (1-x_1^\text{s}) + n_2^\text{s} \right) + C_2^\text{m} \left( \rho (1-x_1^\text{s}) n_2^\text{s} + \delta (1-x_1^\text{s}) n_2^\text{exit} \right),
\end{align}

\noindent Next, we analyze the behavior of the cost functions with respect to \( x_1^\text{s} \):

\begin{itemize}
    \item \( \frac{\partial J_1^\text{s}}{\partial x_1^\text{s}} = C_1^\text{t} \alpha + C_1^\text{m} (\omega n_2^\text{exit} + n_0^\text{enter}) \), which is a positive constant, making \( J_1^\text{s}(\mathbf{x}) \) an increasing affine function of \( x_1^\text{s} \).
    \item \( \frac{\partial J_1^\text{b}}{\partial x_1^\text{s}} = -C_2^\text{t} \gamma - C_2^\text{m} (\rho n_2^\text{s} + \delta n_2^\text{exit}) \), which is a negative constant, making \( J_1^\text{b}(\mathbf{x}) \) a decreasing affine function of \( x_1^\text{s} \).
\end{itemize}

\begin{figure*}[h!]
\centering
\begin{subfigure}{0.32\textwidth}
    \includegraphics[width=\textwidth]{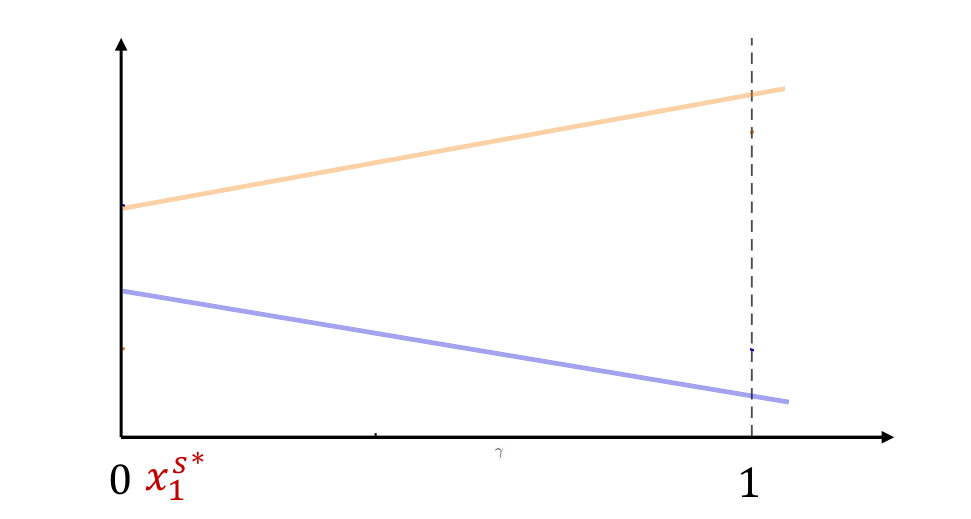}
    \caption{Case(a)}
    \label{fig:HDV_1}
\end{subfigure}
\hfill
\begin{subfigure}{0.325\textwidth}
    \includegraphics[width=\textwidth]{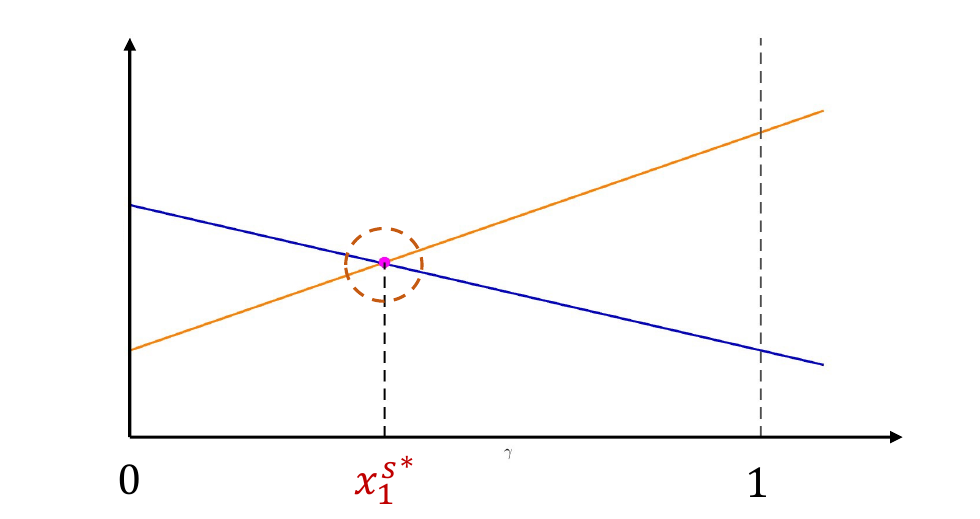}
    \caption{Case(b)}
    \label{fig:HDV_2}
\end{subfigure}
\hfill
\begin{subfigure}{0.325\textwidth}
    \includegraphics[width=\textwidth]{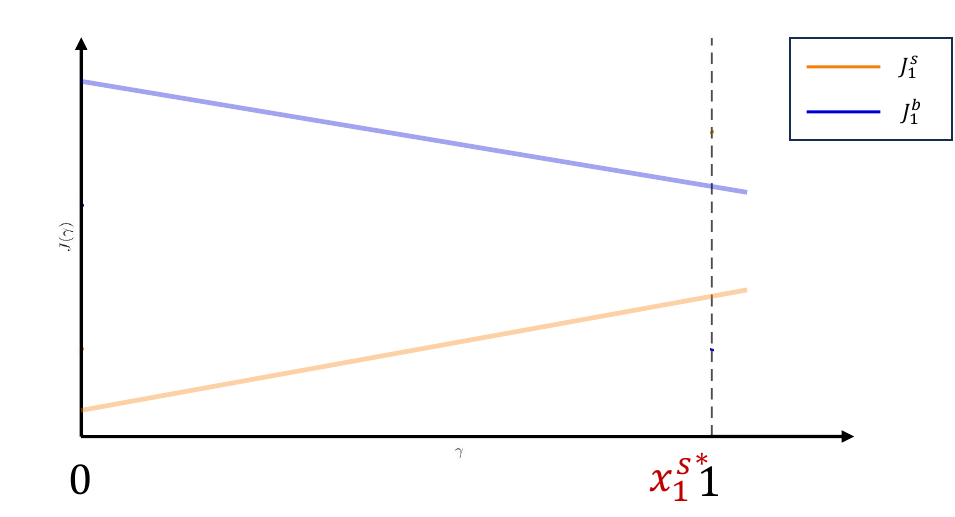}
    \caption{Case(c)}
    \label{fig:HDV_3}
\end{subfigure}

\caption{These figures show three equilibrium cases for HDVs. For Case(a), The bypassing delay cost is consistently higher than the steadfast delay cost across all flow ratios, leading to the equilibrium flow rates $x_1^\text{s}=0$. For Case(b), the bypassing and steadfast delay cost functions intersect, resulting in a mixed equilibrium where both strategies are used: $x_1^\text{b}>0, x_1^\text{s}>0$. For Case(c), The bypassing delay cost is strictly lower than the steadfast delay cost for all flow ratios, yielding an equilibrium of $x_1^\text{s}=1$.}
\label{fig:HDV_eq}
\end{figure*}

The above properties ensure that the cost functions can intersect at most once, which directly implies that the equilibrium is unique. Let us consider three scenarios based on the relative magnitudes of \( J_1^\text{s}(\mathbf{x}) \) and \( J_1^\text{b}(\mathbf{x}) \):

\begin{itemize}
    \item Case (a): For all \( x_1^\text{s} \in [0, 1] \), \( J_1^\text{s}(\mathbf{x}) > J_1^\text{b}(\mathbf{x}) \). This implies that all vehicles will choose to bypass, resulting in \( x_1^\text{s} = 0 \) and \( x_1^\text{b} = 1 \).
    
    \item Case (b): \( J_1^\text{s}(\mathbf{x}) \) and \( J_1^\text{b}(\mathbf{x}) \) intersect at a unique point \( \bar{x}_1^\text{b} \in (0, 1) \). This intersection defines the unique equilibrium distribution, where \( x_1^\text{s} = {x_1^\text{s}}^* \) and \( x_1^\text{b} = 1 - {x_1^\text{s}}^* \).
    
    \item Case (c): For all \( x_1^\text{s} \in [0, 1] \), \( J_1^\text{s}(\mathbf{x}) < J_1^\text{b}(\mathbf{x}) \). In this scenario, all vehicles will remain steadfast, resulting in \( x_1^\text{s} = 1 \) and \( x_1^\text{b} = 0 \).
\end{itemize}

Thus, in all cases, the equilibrium exists and is unique. This concludes the proof.
\end{proof}

\subsection{The Baseline Model Analysis}\label{sec:appendix_model_analysis}

\begin{table*}[!ht]
\caption{Mean Prediction Error Rates (MPER) of the cost function calibrated by the data points in \ref{4A} in different experimental scenarios. The formula for calculating the MPER is: $\text{MPER}=\frac{\sum_{i=1}^{n}\left| \frac{x^\text{s}_{1,\text{S},i}-x^\text{s}_{1,\text{M},i}}{x^{s}_{1,\text{S},i}} \right|}{n}\times 100\%$, in which $n$ is the length of test set; $x^{\text{s}}_{1,\text{S},i}$ is the $i^{\text{th}}$ simulation-generated $x^{\text{s}}_{1}$ in the test set; $x^{\text{s}}_{1,\text{M},i}$ is the $i^{\text{th}}$ $x^{\text{s}}_{1}$ predicted by our model. ${\text{MPER}^{0.1667}_{n_i}}$ represents the experiments with its constant flow \( n_i=0.1667\) for \( i \in \{\text{0, 2}\} \), while ${\text{MPER}^{0.4167}_{n_i}}$ represents the experiments with its constant flow \( n_i=0.4167\) for \( i \in \{\text{0, 2}\} \)}
\centering
\begin{tabular}{c >{\centering\arraybackslash}p{3cm} >{\centering\arraybackslash}p{3cm} >{\centering\arraybackslash}p{3cm} >{\centering\arraybackslash}p{3cm}}
\hline

Expt.   & ${\text{MPER}^{0.1667}_{n_0^\text{{enter}}}}$ & ${\text{MPER}^{0.1667}_{n_2^\text{{s}}}}$   & ${\text{MPER}^{0.4167}_{n_0^\text{{enter}}}}$& ${\text{MPER}^{0.4167}_{n_2^\text{{s}}}}$ \\ \hline
\ref{4A}  & 1.15\% & 1.55\% & 1.00\% & 1.05\% \\ 
1-1 & 5.51\% & 3.52\% & 2.41\% & 4.79\% \\ 
1-2 & 5.89\% & 11.00\% & 8.36\% & 1.85\% \\ 
2-1 & 2.54\% & 3.07\% & 3.18\% & 4.30\% \\ 
2-2 & 1.79\% & 2.23\% & 2.74\% & 1.18\% \\ 
3-1 & 1.13\% & 1.26\% & 1.21\% & 0.72\% \\ 
3-2 & 1.15\% & 1.36\% & 1.01\% & 0.87\% \\ 
3-3 & 1.04\% & 1.36\% & 1.10\% & 0.96\% \\ 
\hline
                         
\end{tabular}
\label{table1}
\end{table*}

\noindent To test the robustness of our model, we conduct the following three sets of univariate experiments:
\begin{enumerate}
    \item Adjust the velocity limits of through vehicles on lane 2 from 20 m/s to 12.5 m/s (Expt. 1-1) and the velocity limits of exiting vehicles from 20 m/s to 13.9 m/s (Expt. 1-2).
    \item Adjust the minimum gap (the minimum distance a vehicle can maintain from the vehicle ahead) of through vehicles on lane 2 (Expt. 2-1) and entering vehicles (Expt. 2-2) from 2m to 10m.
    \item Increase the vehicle aggressiveness level of through vehicles on lane 2 (Expt. 3-1),  exiting vehicles (Expt. 3-2) and entering vehicles (Expt. 3-3)
\end{enumerate}

We input the data points collected from the experiments into the cost function calibrated by the data points in \ref{4A}. For each test, we set either $n_0^\text{enter}$ or $n_2^\text{s}$ as constant, with values of \( n_0^\text{enter},n_2^\text{s} \in \left\{0.1667,0.4167\right\} \), following the method used in the model validation section to select data points. From each test, we obtain 100 data points: 50 data points are collected when $n_0^\text{enter}$ is constant (25 data points for $n_0^\text{enter}=0.1667$ and 25 data points for $n_0^\text{enter}=0.4167$), and the remaining 50 data points are collected while $n_2^\text{s}$ is constant. As the table \ref{table1} shows, all of the mean prediction error rates in different experimental scenarios are no larger than 11\%, which shows that the model is robust.

To explore the practical significance implied by each parameter in the model, thereby enabling informed adjustments under different environmental configurations to maintain accurate model performance, we recalibrate our model using the configuration adjusted in the experiments. The results are shown in Table \ref{table2} and Table \ref{table3}.

\begin{table}[!ht]
\caption{Model parameters under different configurations.}
\centering
%\begin{tabular}{ccccccc}
\begin{tabular}{c >{\centering\arraybackslash}p{2cm} >{\centering\arraybackslash}p{2cm} >{\centering\arraybackslash}p{2cm} >{\centering\arraybackslash}p{2cm} >
{\centering\arraybackslash}p{2cm} >{\centering\arraybackslash}p{2cm}}
\hline
Expt.    & $\alpha$   & $\beta$  & $\omega$ & $\gamma$  & $\rho$                  & $\delta$        \\ \hline
\ref{4A}            & 1.255    & 1.138   & 1.000           & 2.384     & 1.000            & 3.094    \\ 
1-1              & 1.323    & 2.618   & 1.000           & 2.323 & 1.000            & 6.240  \\ 
1-2 & 1.178 & 2.002 & 1.000 & 2.116 & 1.000 & 8.266\\ 
2-1        & 1.000 & 1.000 & 1.030 & 2.323 & 1.123 & 2.459\\ 
2-2           & 1.236 & 1.000 & 1.000 & 2.204 & 1.044 & 2.726 \\ 
3-1    & 1.302 & 1.056 & 1.000 & 2.436 & 1.000 & 2.958\\ 
%EXP 3-2        & 1.424 & 1.000 & 1.527 & 3.479 & 1.000 & 2.642 \\ 
3-2         & 1.324 & 1.000 & 1.000 & 2.475 & 1.000 & 2.835  \\ 
3-3    & 1.294 & 1.000 & 1.076 & 2.417 & 1.000 & 2.985\\ 
\hline
                           
\end{tabular}
\label{table2}
\end{table}

\begin{table*}[!ht]
\caption{The fluctuating rate (FR) between model parameters calibrated under different configurations and the configuration we used in \ref{4A} . The formula for calculating the fluctuating rate is: $\text{FR}=\frac{x_{\text{exp}}-x_{\text{ori}}}{x_\text{ori}}\times100\%$. If $\text{FR}>0$, the parameters obtained from the experiments are larger than the corresponding parameters we calibrated in \ref{4A}. Otherwise, if $\text{FR}<0$, the parameters obtained from the experiments are smaller than the corresponding parameters from \ref{4A}. The bold data represent parameters that have significant differences ($\left| \text{FR} \right|>25\%$) compared to the parameters from \ref{4A}. }
\centering
%\begin{tabular}{ccccccc}
\begin{tabular}{c >{\centering\arraybackslash}p{2cm} >{\centering\arraybackslash}p{2cm} >{\centering\arraybackslash}p{2cm} >{\centering\arraybackslash}p{2cm} >
{\centering\arraybackslash}p{2cm} >{\centering\arraybackslash}p{2cm}}
\hline
Expt.       & $\alpha$    & $\beta$  & $\omega$ & $\gamma$  & $\rho$    & $\delta$  \\ \hline
1-1     & 5.42\%    & \textbf{130.05\%}   & 0.00\%   & -2.56\% & 0.00\%  & \textbf{101.68\%}  \\ 
1-2     & -6.13\% & \textbf{75.92\%} & 0.00\% & -11.24\% & 0.00\% & \textbf{167.16\%}\\ 
2-1     & -20.32\% & -12.13\% & 3.00\% & -2.56\% & 12.30\% & -20.52\%\\ 
2-2     & -1.51\% & -12.13\% & 0.00\% & -7.55\% & 4.40\% & -11.89\% \\ 
3-1     & 3.75\% & -7.21\% & 0.00\% & 2.18\% & 0.00\% & -4.40\% \\ 
%EXP 3-2     & 13.47\% & -12.13\% & \textbf{52.70\%} & \textbf{45.93\%} & 0.00\% & -14.61\% \\ 
3-2     & 5.50\% & -12.13\% & 0.00\% & 3.82\% & 0.00\% & -8.37\%  \\ 
3-3     & 3.11\% & -12.13\% & 7.60\% & 1.38\% & 0.00\% & -3.52\%\\ 
\hline
                           
\end{tabular}
\label{table3}
\end{table*}

The results of experiments show that parameters $\omega$ and $\rho$ have very strong robustness since they gained a 0\% fluctuating rate in most of the experiments, while they only obtained a minimal fluctuating rate in some specific experiments. Parameters $\alpha$ and $\gamma$ also performed robustly, keeping the absolute value of the fluctuating rate below 25\% in all experiments. Parameters $\beta$ and $\delta$ acted robustly in experiments 2 and 3, as all the absolute values of the fluctuating rates are below 25\% in both experiments. Meanwhile, both values gained a high fluctuating rate when the velocity of vehicles in the scenario changed. Given that the results show the changing trends of different parameters under various factors, the model parameters can be adjusted accordingly in new experimental environments to achieve more accurate prediction performance.

In conclusion, we validated the accuracy of our model and further assessed its robustness across different experimental scenarios. By recalibrating the model under various experimental conditions, we identified the variation patterns of model parameters across scenarios with different configurations. These findings enable the model to achieve accurate prediction performance by adjusting the corresponding parameters when applied to unprecedented weaving ramp scenarios not included in our training data. Therefore, our model demonstrates strong practicality when applied to previously unseen scenarios.

\subsection{Proof for Theorem 2}\label{sec:appendix_theorem2}
\begin{proof}
In this problem, the decision variable is steadfast proportion $q_\text{s} \in [0,1]$. The CAV penetration rate $p \in [0,1]$ is given.  From Equation \eqref{eq:x_1^s} and \eqref{eq:choice_CAV_s}, the range of CAV steadfast proportion $x_{1,\text{CAV}}^\text{s} \in [0,p]$, and the range of HDV steadfast proportion $x_{1,\text{HDV}}^\text{s} \in [0,1-p]$. 

In Appendix \ref{sec:appendix_hdv_eq}, the existence and uniqueness of HDV equilibrium is proved. In Equation \eqref{eq:wardrop_lower}, the lower level is the HDV equilibrium, thus the lower level satisfies the uniqueness and existence. Then we prove the existence and uniqueness of upper level CAV optimization:

In Equation \eqref{eq:upper_xb}, $q_\text{s}^*$ is the point that can let $J_\text{soc}$ reaches the minimum value. Thus, in order to prove the existence and uniqueness of $q_\text{s}^*$, the existence and uniqueness of minimum $J_\text{soc}$ needs to be proved. The social cost $J_\text{soc}$ can be expressed as a function of the total steadfast proportion $x_1^\text{s}$:
\begin{equation}\label{eq:social_cost_beta}
  \begin{aligned}
    J_\text{soc}
    & = (K_1^\text{s}+K_1^\text{b})({x_1^\text{s}})^2+(-2K_1^\text{b}+n_2^\text{exit}K_2^\text{exit} + n_0^\text{enter}K_0^\text{enter} + B_1^\text{s}-n_2^\text{s} K_2^\text{s} - B_1^\text{b})x_1^\text{s} \\
    & +K_1^\text{b}+n_2^\text{s} K_2^\text{s} + B_1^\text{b}+B_\text{soc},
  \end{aligned}
\end{equation}

%&= K_1^\text{s} ({x_{1}^\text{s}})^2+ K_1^\text{b} (1-{x_{1}^\text{s}})^2+(n_2^\text{exit}K_2^\text{exit} + n_0^\text{enter}K_0^\text{enter} + B_1^\text{s})x_1^\text{s}+(n_2^\text{s} K_2^\text{s} + B_1^\text{b})(1-x_1^\text{s})+B_\text{soc}, \\

This expression defines a polynomial in $x_1^\text{s}$, where $x_1^\text{s} \in [0,1]$ lies in a compact and convex domain. The cost function $J_{\text{soc}}$ is continuous and differentiable on this interval, and reduces to a quadratic form. The leading coefficient of the quadratic term in Equation~\eqref{eq:social_cost_beta} is $K_1^\text{s} + K_1^\text{b}>0$. Hence, $J_\text{soc}$ is strictly convex over the domain of $x_1^\text{s}$. In Equation \eqref{eq:x_1^s} and \eqref{eq:choice_CAV_s}, $x_1^\text{s}=x_{1,\text{HDV}}^\text{s}+pq_\text{s}$. As $x_{1,\text{HDV}}^\text{s*}$ is unique and $p$ is given, ${q_\text{s}}^*$ exists and is unique for a minimum value of $J_\text{soc}$. 

Based on above convex analysis, $J_\text{soc}$ has a global minimum point $x_1^\text{s}$, which we defined as $\mathbf{B}$:
\begin{align}
\mathbf{B}=\frac{2K_1^\text{b} + B_1^\text{b} - B_1^\text{s} -n_2^\text{exit}K_2^\text{exit} - n_0^\text{enter}K_0^\text{enter}+ n_2^\text{s} K_2^\text{s}}{2(K_1^\text{s}+K_1^\text{b})},
\end{align}
%\begin{align}
%\mathbf{B}=\frac{2K_2-K_3+K_4}{2(K_1+K_2)},
%\end{align}
where we have $\mathbf{B}>\Phi$, as we define the social optimal scenario has higher steadfast proportion than HDV-only scenario. 

Let's consider possible equilibrium cases (denoted as $\mathbf{x}^*$) at mixed autonomy scenario. Given CAV penetration rate $p$, there are 3 possible cases, which are: Case 1. $J_1^\text{s}(\mathbf{x}^*)=J_1^\text{b}(\mathbf{x}^*)$. Case 2. $J_1^\text{s}(\mathbf{x}^*)>J_1^\text{b}(\mathbf{x}^*)$. Case 3. $J_1^\text{s}(\mathbf{x}^*)<J_1^\text{b}(\mathbf{x}^*)$. 

\textbf{Case 1}: 

If $J_1^\text{s}(\mathbf{x}^*)=J_1^\text{b}(\mathbf{x}^*)$ exists: based on Equation \eqref{eq:J_1^s_x_1^s} and  \eqref{eq:J_1^b_x_1^b} , together with $x_1^\text{s}+x_1^\text{b}=1$, the optimal ${x_1^\text{s}}^*$ is calculated as:
\begin{align}
  {x_1^\text{s}}^*=\frac{K_1^\text{b}+B_1^\text{b}-B_1^\text{s}}{K_1^\text{s}+K_1^\text{b}},
\end{align}
where is the same steadfast proportion as HDV-only equilibrium $\Phi$. Thus, ${x_1^\text{s}}^*=\Phi$ in this case, and $\Phi \in [0,\mathbf{B}]$.

\begin{figure}[h!]
\centering
\includegraphics[width = 0.5\textwidth]{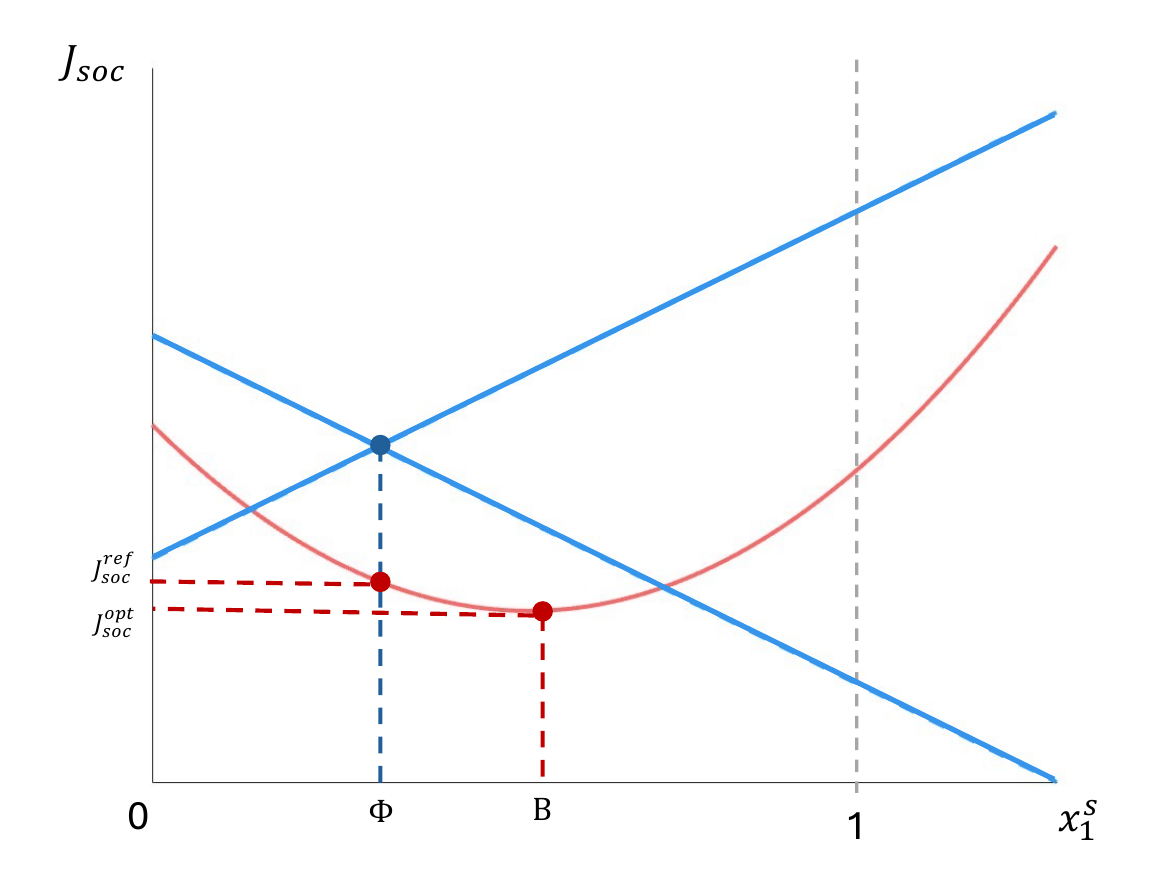}
\caption{This figure shows three equilibrium cases for mixed autonomy scenario.}
\end{figure}

Based on Equation \eqref{eq:choice_CAV_s}, the CAV proportion $x_{1,\text{CAV}}^\text{s}=pq_\text{s}$, and the corresponding HDV steadfast proportion $x_{1,\text{HDV}}^\text{s}$ can be updated as:
\begin{align}
  x_{1,\text{HDV}}^\text{s}=\Phi-pq_\text{s},
\end{align}

Due to $x_{1,\text{HDV}}^\text{s}\in [0,1-p]$, we have:
\begin{align} \label{eq:x_hdv_s_range}
  0 \leq \Phi - p q_\text{s} \leq 1-p
\end{align}

As $\Phi<\mathbf{B}$, all CAV prefer steadfast to minimize $J_\text{soc}$. Therefore, under ${q_\text{s}}^*=1$, Equation~\eqref{eq:x_hdv_s_range} must hold. Since $\Phi \leq 1$ always holds, the admissible range of $p$ that allows this case to occur is:
\begin{align}
  p\leq \Phi
\end{align}

Thus, when $0 \leq p \leq \Phi$, the equilibrium case satisfies $J_1^\text{s}(\mathbf{x}^*) = J_1^\text{b}(\mathbf{x}^*)$. When $p$ increases, $J_\text{soc}$ keeps invariant, and $J_\text{soc}=J_\text{soc}^\text{ref}$. 

%\begin{align}
%  \begin{cases}
%    \mathbf{A} + p - 1 \leq p q_\text{s}^* \leq \mathbf{A}, \\
%    0 \leq p q_\text{s}^*  \leq p,
%  \end{cases}
%\end{align}

\iffalse
1) When $ 0 \leq \mathbf{A}-pq_\text{s} \leq 1-p$, $x_{1,\text{CAV}}^\text{s}=pq_\text{s}$ should find valid interval to get valid result:
\begin{align}
  \begin{cases}
    \mathbf{A} + p - 1 \leq p q_\text{s} \leq \mathbf{A}, \\
    0 \leq p q_\text{s} \leq p,
  \end{cases}
\end{align}

As $\mathbf{A}+p-1 \leq p$ is always true, $x_{1,\text{CAV}}^\text{s}$ can always find a valid value to let ${x_1^\text{s}}^*=\mathbf{A}$. $J_\text{soc}=(K_1+K_2)\mathbf{A}^2+(K_3-2K_2-K_4)\mathbf{A}+K_2+K_4+B$, which is a fixed value.

2) When $ \mathbf{A}-pq_\text{s} < 0$ or $ \mathbf{A}-pq_\text{s} > p$: If $ \mathbf{A}-pq_\text{s}<0$, $x_{1,\text{HDV}}^\text{s}<0$. If $ \mathbf{A}-pq_\text{s} > p$, $x_{1,\text{HDV}}^\text{s}>1$. Both break the rule $ 0 \leq x_{1,\text{HDV}}^\text{s} \leq 1$, and there is no valid $x_{1,\text{HDV}}^\text{s}$ under such situations.
\fi

\textbf{Case 2}:

If $J_1^\text{s}(\mathbf{x}^*)>J_1^\text{b}(\mathbf{x}^*)$ exists: ${x_{1,\text{HDV}}^\text{s}}^*=0$. From Equation \eqref{eq:x_1^s}, it follows that ${x_{1,\text{CAV}}^\text{s}}^*={x_1^\text{s}}^*$. Based on Equation \eqref{eq:J_1^s_x_1^s} and \eqref{eq:J_1^b_x_1^b}, together with $x_1^\text{s}+x_1^\text{b}=1$, the optimal ${x_1^\text{s}}^*$ satisfies:
\begin{align}
  {x_1^\text{s}}^*>\frac{K_1^\text{b}+B_1^\text{b}-B_1^\text{s}}{K_1^\text{s}+K_1^\text{b}}=\Phi,
\end{align}

There are two possible minimum cases for $J_\text{soc}$: (1) $\Phi<{x_1^\text{s}}^*< \mathbf{B}$, (2) ${x_1^\text{s}}^* \geq \mathbf{B}$.

(1) $\Phi<{x_1^\text{s}}^*< \mathbf{B}$. In this case, ${x_1^\text{s}}={x_{1,\text{CAV}}^\text{s}}=pq_\text{s}$, and all CAV prefer steadfast, which means $q_\text{s}^*=1$, thus ${x_1^\text{s}}^*=p$. Also, the range of $p$ is updated to be $\Phi<p < \mathbf{B}$.

From Equation \eqref{eq:social_cost_beta}, $J_\text{soc}$ can be updated to be a quadratic function of $p$, and the range of p is $\Phi<p< \mathbf{B}$. When $p\in (\Phi,\mathbf{B})$ increases, $J_\text{soc}$ decreases.

(2) ${x_1^\text{s}}^* \geq \mathbf{B}$. In this case, $x_1^\text{s}=pq_\text{s}\geq \mathbf{B}$. As the range of $x_1^\text{s}\in [0,p]$, if this case exists, there must satisfy $p\geq \mathbf{B}$. 

From Equation \eqref{eq:social_cost_beta}, $J_\text{soc}$ is a quadratic function of $x_1^\text{s}$, and the range of $x_1^\text{s}\geq \mathbf{B}$. $J_\text{soc}$ always reach its minimum value at ${x_1^\text{s}}^*=\mathbf{B}$. When $p\in [\mathbf{B},1]$ increases, $J_\text{soc}$ keeps invariant.

%Based on Equation \eqref{eq:x_1^s}, $x_{1,\text{CAV}}^\text{s}={x_1^\text{s}}^*$. As $x_{1,\text{CAV}}^\text{s} \in [0,p]$, and ${x_1^\text{s}}^*\in [\Phi,1]$. If this case $J_1^\text{s}(\mathbf{x}^*)>J_1^\text{b}(\mathbf{x}^*)$ exists, the range of $x_{1,\text{CAV}}^\text{s}$ and ${x_1^\text{s}}^*$ should have overlap. Thus in this case, the range of $p$ should satisfy:
%\begin{align}
%  p>\Phi
%\end{align}

%\begin{align}
%{x_1^\text{s}}_{min}=\frac{2K_2-K_3+K_4}{2(K_1+K_2)},
%\end{align}
%where we define $\frac{2K_2-K_3+K_4}{2(K_1+K_2)}=\mathbf{B}$. In our previous analysis, it is assumed that social optimal situation has higher proportion of $x_1^\text{s}$ compared with Wardrop equilibrium situation, thus we have $\mathbf{A} \leq \mathbf{B}$, and $\mathbf{B}\geq 0$.

%1) When $\mathbf{B}\in [0,p]$, ${x_1^\text{s}}_{min}={x_1^\text{s}}^*=x_{1,\text{CAV}}^\text{s}=\mathbf{B}$. In this case, $p\geq \mathbf{B}$, ${x_1^\text{s}}^*=\mathbf{B}$, and $J_\text{soc}=K_2+K_4+B$, which is a fixed value.

%2) When $\mathbf{B} > p$, $x_{1,\text{CAV}}^\text{s}={x_1^\text{s}}^*=p$.  In this case, $\mathbf{A} < p < \mathbf{B} $, and $J_\text{soc}=(K_1+K_2)p^2+(-2K_2+K_3+K_4)p+K_2+K_4+B$. Because $\mathbf{B}={x_1^\text{s}}_{min}>p$, when $p$ increases, $J_\text{soc}$ decreases. 

%Overall, in Case 2, in $\mathbf{A}<p<\mathbf{B}$, when $p$ increases, $J_\text{soc}$ decreases. In $p\geq\mathbf{B}$, when $p$ increases, $J_\text{soc}=K_2+K_4+B$ reaches the global minimum value and keep invariant.

\textbf{Case 3}:

If $J_1^\text{s}(\mathbf{x}^*)<J_1^\text{b}(\mathbf{x}^*)$ exists: $x_{1,\text{HDV}}^\text{s}=1-p$, and $x_1^\text{s}=1-p+pq_\text{s}$. Based on Equation \eqref{eq:J_1^s_x_1^s} and \eqref{eq:J_1^b_x_1^b}, together with $x_1^\text{s}+x_1^\text{b}=1$, the optimal ${x_1^\text{s}}^*$ satisfies:
\begin{align}
  {x_1^\text{s}}^*<\frac{K_1^\text{b}+B_1^\text{b}-B_1^\text{s}}{K_1^\text{s}+K_1^\text{b}}=\Phi,
\end{align}

Since $x_1^\text{s}<\Phi$, all CAV prefer steadfast, thus $q_\text{s}=1$. Also, we have $1-p+pq_\text{s}<\Phi$, thus $\Phi>1$. This contradicts the range of $\Phi<\mathbf{B}\leq 1$, so this case does not exists. 

%Because $J_1^\text{s}(\mathbf{x}^*)<J_1^\text{b}(\mathbf{x}^*)$, $x_{1,\text{HDV}}^\text{s}=1-p$. Then $x_1^\text{s}=1-p+pq_\text{s}<\mathbf{A}\leq \mathbf{B}$. Every $J_\text{soc}$ in this case is bigger than Case 1 and Case 2.

In conclusion, combined Case 1, Case 2, and Case 3: in $0\leq p \leq \mathbf{A}$, when p increases, ${x_1^\text{s}}^*=\Phi$, $J_\text{soc}$ keep invariant. In $\Phi < p < \mathbf{B}$, when $p$ increases, $J_\text{soc}$ decreases. In $\mathbf{B}\leq p\leq1$, when $p$ increases, $J_\text{soc}$ reaches the global minimum point and keep invariant.

\end{proof}

\subsection{Proof for Theorem 3} \label{sec:appendix_theorem3}
\begin{proof}

%All CAV and HDV types are described as $j$ in the mixed autonomy scenario, which we have a disjoint union $\mathcal{J} = \mathcal{H} \cup \mathcal{C} $. Based on Equations \eqref{eq:SVO_cost1} and \eqref{eq:SVO_cost2}, and noting that the total steadfast proportion is given by $x_1^\text{s}=\sum_{j\in\mathcal{J}} x_{1,j}^\text{s}$, i.e., the aggregate steadfast ratio equals the sum of the steadfast proportions of all vehicle types, we can reformulate the equations as follows:

In the mixed-autonomy setting, all vehicle types $j \in \mathcal{J} = \mathcal{H} \cup \mathcal{C}$ are treated within a unified framework, as established in Section~\ref{sec:svo}. Recall that the total steadfast proportion satisfies $x_1^\text{s} = \sum_{j \in \mathcal{J}} x_{1,j}^\text{s}$, so that $\frac{\partial x_1^\text{s}}{\partial x_{1,j}^\text{s}} = 1$ and likewise $\frac{\partial x_1^\text{b}}{\partial x_{1,j}^\text{b}} = 1$ for all $j$. Substituting these relations into Equations~\eqref{eq:SVO_cost1} and~\eqref{eq:SVO_cost2}, the type-specific cost functions simplify to:
\begin{align}
  \tilde{J}_{1,j}^\text{s}(\mathbf{x}) &= cos(\theta_j)J_{1}^\text{s}(\mathbf{x})+ sin(\theta_j)\frac{\partial J_\text{soc}(\mathbf{x})}{\partial x_{1}^\text{s}} \frac{\partial}{\partial x_{1,j}^\text{s}}\sum_{j=1}^J x_{1,j}^\text{s},\\
  &=(cos(\theta_j)+2sin(\theta_j))K_1^\text{s}x_1^\text{s} + cos(\theta_j)B_1^\text{s} + sin(\theta_j)(B_1^\text{s} + n_2^\text{exit}K_2^\text{exit} + n_0^\text{enter}K_0^\text{enter}), \label{eq:cost_func_J_1js} \\
  \tilde{J}_{1,j}^\text{b}(\mathbf{x}) &= cos(\theta_j)J_1^\text{b}(\mathbf{x}) + sin(\theta_j)\frac{\partial J_\text{soc}(\mathbf{x})}{\partial x_{1}^\text{b}} \frac{\partial }{\partial x_{1,j}^\text{b}} \sum_{j=1}^J x_{1,j}^\text{b},\\
  &= (cos(\theta_j)+2sin(\theta_j)) K_1^\text{b} x_1^\text{b} + cos(\theta_j) B_1^\text{b}+ sin(\theta_j)( B_1^\text{b} + n_2^\text{s} K_2^\text{s}), \label{eq:cost_func_J_1jb}
\end{align}
Since $x_1^\text{s}=\sum_j x_{1,j}^\text{s}$, we have $\frac{\partial{x_1^\text{s}}}{\partial x_{1,j}^\text{s}}$, hence $\frac{\partial{J_\text{soc}}}{\partial x_{1,j}^\text{s}}=\frac{\partial{J_\text{soc}}}{\partial x_{1}^\text{s}}$. Same for $x_1^\text{b}$. Therefore, we have the fact that $\frac{\partial}{\partial x_{1,j}^\text{s}}\sum_{j=1}^J x_{1,j}^\text{s}=1$ and $\frac{\partial}{\partial x_{1,j}^\text{b}}\sum_{j=1}^J x_{1,j}^\text{b}=1$.

Also, we can simplify Equation~\eqref{eq:cost_func_J_1js} and~\eqref{eq:cost_func_J_1jb} as:
\begin{align}
  \tilde{J}_{1,j}^\text{s}&=K_{1,j}^\text{s} x_1^\text{s}+B_{1,j}^\text{s},\\
  \tilde{J}_{1,j}^\text{b}&=K_{1,j}^\text{b}x_1^\text{b}+B_{1,j}^\text{b}.
\end{align}
where $K_{1,j}^\text{s}:=(cos(\theta_j)+2sin(\theta_j))K_1^\text{s}$, $B_{1,j}^\text{s}:=cos(\theta_j)B_1^\text{s} + sin(\theta_j)(B_1^\text{s} + n_2^\text{exit}K_2^\text{exit} + n_0^\text{enter}K_0^\text{enter})$, $K_{1,j}^\text{b}:=(cos(\theta_j)+2sin(\theta_j)) K_1^\text{b}$, and $B_{1,j}^\text{b}:=cos(\theta_j) B_1^\text{b}+ sin(\theta_j)( B_1^\text{b} + n_2^\text{s} K_2^\text{s})$.

We define $\Delta_j(x_1^\text{s})=\tilde{J}_{1,j}^\text{s}-\tilde{J}_{1,j}^\text{b}$, which represents the cost difference between choosing steadfast and the bypass for vehicle type $j$. Since the total proportion on the two behaviors satisfies $x_1^\text{s}+x_1^\text{b}=1$, we can substitute $x_1^\text{b}=1-x_1^\text{s}$ into the above expression. Thus, $\Delta_j(x_1^\text{s})$ becomes a function of $x_1^\text{s}$ only and can be written in the following linear form:
\begin{align}
  \Delta_j(x_1^\text{s})= (K_{1,j}^\text{s}+K_{1,j}^\text{b})x_1^\text{s}+B_{1,j}^\text{s}-B_{1,j}^\text{b}-K_{1,j}^\text{b}.
\end{align}

Since $\theta_j\in (-arctan\frac{1}{2}, \pi-arctan\frac{1}{2})$, we have $cos(\theta_j)+2sin(\theta_j)>0$. Under this condition, for any vehicle type $j$, $J_{1,j}^\text{s}$ is monotonically increasing in $x_1^\text{s}$, whereas $J_{1,j}^\text{b}$ is monotonically decreasing in $x_1^\text{s}$. Assume there exists type $j$'s equilibrium $\tilde{J}_{1,j}^\text{s}(\chi_j)=\tilde{J}_{1,j}^\text{b}(\chi_j)$, which we define the total steadfast proportion $x_1^\text{s}$ at such equilibrium is $\chi_j$:
\begin{align}
  \chi_j=\frac{B_{1,j}^\text{b}+K_{1,j}^\text{b}-B_{1,j}^\text{s}}{K_{1,j}^\text{s}+K_{1,j}^\text{b}}
\end{align}

\iffalse
\begin{align}
\chi_j=-\frac{cos(\theta_j)(B_1^\text{s}-K_1^\text{b}-B_1^\text{b})+sin(\theta_j)(-2K_1^\text{b}+B_1^\text{s}+n_2^\text{exit}K_2^\text{exit}+n_0^\text{enter}K_0^\text{enter}-B_1^\text{b}-n_2^\text{s}K_2^\text{s})}{(cos(\theta_j)+2sin(\theta_j))(K_1^\text{s}+K_1^\text{b})}.
\end{align}
\fi

For the true ${x_1^\text{s}}^*$ under penetration $p$, vehicle type $j$ can only fall into one of the following three cases:
\begin{itemize}
  \item If ${x_1^\text{s}}^*<\chi_j$, we have $\tilde{J}_{1,j}^\text{s}<\tilde{J}_{1,j}^\text{b}$, type $j$ strictly keeps steadfast ($x_{1,j}^\text{s}=w_j(p)$).
  
  \item If ${x_1^\text{s}}^*>\chi_j$, we have $\tilde{J}_{1,j}^\text{s}>\tilde{J}_{1,j}^\text{b}$, type $j$ strictly keeps bypass ($x_{1,j}^\text{s}=0$).
  
  \item If ${x_1^\text{s}}^*=\chi_j$, we have $\tilde{J}_{1,j}^\text{s}=\tilde{J}_{1,j}^\text{b}$, type $j$ is indifferent between two choices.
\end{itemize}

Since all vehicle types are distinct and each type has a typical $\theta_j$, which means each vehicle type has a unique threshold $\chi_j$. Then we rank all thresholds $\chi_j$ in ascending order:
\begin{align}
  \chi_{1}<\chi_{2}<\cdots<\chi_{J}.
\end{align}

%For any $x_1^\text{s}\in [0,1]$, define the set $M=\{j | C_j\geq x_1^\text{s}\}$, which includes all vehicle types that fully or partial steadfast.

In any Wardrop equilibrium, there exists an index $k$ such that:
\begin{itemize}
    \item all types with threshold $> \chi_{k}$ are fully steadfast.
    
    \item all types with threshold $< \chi_{k}$ are fully bypass.
    
    \item only type $k$ may mix (if ${x_1^\text{s}}^*=\chi_{k}$).
\end{itemize}

%For every $ j\in M$, the steadfast proportion $x_{1,j}^\text{s}$ is given by:
%\begin{align}
%  \begin{cases}
%    x_{1,j}^\text{s}=w_j, &\quad if \quad C_j>x_1^\text{s} \\
%    x_{1,j}^\text{s}=C_j-\sum_{j'\in M,j'\neq j} w_{j'} &\quad if \quad C_j=x_1^\text{s}, 
%  \end{cases}
%\end{align}
%where $\sum_{j\in M} x_{1,j}^\text{s} = x_1^\text{s}$.

\textbf{Check the monotonicity of $W_k(p)$ for $p$ under fixed $k$. }

Define total steadfast proportion tail sums other than vehicle $k$ as $W_k(p)$ for $k=1,\cdots,J$:
\begin{align}
    W_k(p)= \sum_{j\in \mathcal{H},\chi_j>\chi_{k}}(1-p)w_h^\text{HDV} + \sum_{j\in \mathcal{C}, \chi_j>\chi_{k}} p w_c^\text{CAV} \label{eq:W_k(p)}
\end{align}

When $k$ is fixed, Equation~\eqref{eq:W_k(p)} is affine in $p$. Let 
\begin{align}
A_k := \sum_{j\in \mathcal{C},\chi_j>\chi_k} w_c^\text{CAV} - \sum_{j\in \mathcal{H},\chi_j>\chi_k} w_h^\text{HDV}.
\end{align}

Then $W_k(p)$ has slope $A_k$ with respect to $p$. Therefore, $W_k(p)$ is increasing, constant, or decreasing in $p$ 
if $A_k > 0$, $A_k = 0$, or $A_k < 0$, respectively.

\textbf{Check the monotonicity of $J_\text{soc}$ for $p$ under fixed $k$.}

Under fixed vehicle type $k$, the total steadfast proportion ${x_1^\text{s}}^*=\chi_k$, which is fixed. According to Equation~\eqref{eq:social_cost_beta}, the social cost $J_\text{soc}$ depends only on $x_1^\text{s}$. As a result, $J_\text{soc}$ remains constant within the region where the vehicle type $k$ is fixed. Consequently, variations in $J_\text{soc}$ occur only at the switching points where the fixed vehicle type $k$ changes.

%When $k$ is fixed, Equation~\ref{eq:W_k(p)} is a linear function of $p$. The monotonicity of $W_k(p)$ with respect to $p$ is is determined by the sign of its coefficient $\sum_{j\in \mathcal{C},\chi_j>\chi_k}w_c^\text{CAV}-\sum_{j\in \mathcal{H},\chi_j>\chi_k} w_h^\text{HDV}$. When coefficient $>0$, $W_k(p)$ is monotonically increasing with $p$; when coefficient $=0$, $W_k(p)$ is constant; when coefficient $<0$, $W_k(p)$ is monotonically decreasing with $p$.

%{\color{red}As $w_h^{HDV}$ and $w_c^{CAV}$ are constant vectors, and CAVs are defined to be more beneficial to society compared with HDVs. At least as much weight of CAVs has high steadfast thresholds as among HDVs, which means $\sum_{j\in \mathcal{C}, \chi_j>\chi_{k}} w_c^\text{CAV} > \sum_{j\in \mathcal{H},\chi_j>\chi_{k}}w_h^{HDV}$. Hence $W_k(p)$ is a linear and monotonically increasing function of $p$ under fixed $k$.}

\textbf{Classify type $k$'s behavior.}

When $k$ is fixed, there are exactly two behavior regimes for vehicle type $k$:

\textbf{Regime 1: Interior Regime for Mixed type $k$.} When vehicle type $k$ satisfy: 
\begin{align}
 0<\chi_k-W_{k}(p)<w_{k}(p), \label{eq:regime1}
\end{align}
Vehicle type $k$ adopts a mixed strategy and is partially steadfast and partially bypass (Figure~\ref{fig:mixed_k}). Vehicle types with $\chi_j > \chi_k$ are fully steadfast, while vehicle types with $\chi_j < \chi_k$ are fully bypass.

\begin{figure*}[h!]
\centering
\begin{subfigure}{0.32\textwidth}
    \includegraphics[width=\textwidth]{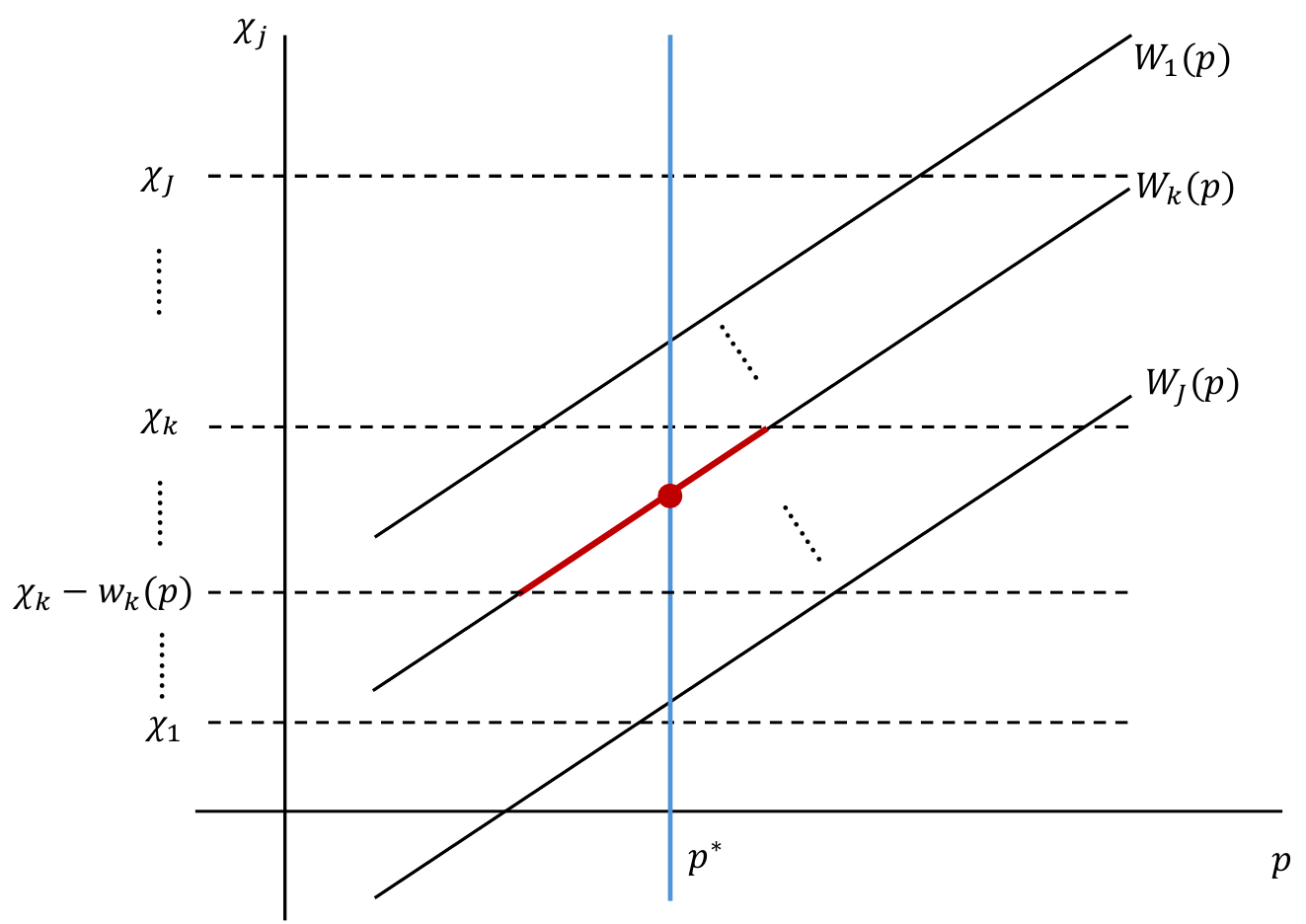}
    \caption{Mixed type $k$}
    \label{fig:mixed_k}
\end{subfigure}
\hfill
\begin{subfigure}{0.325\textwidth}
    \includegraphics[width=\textwidth]{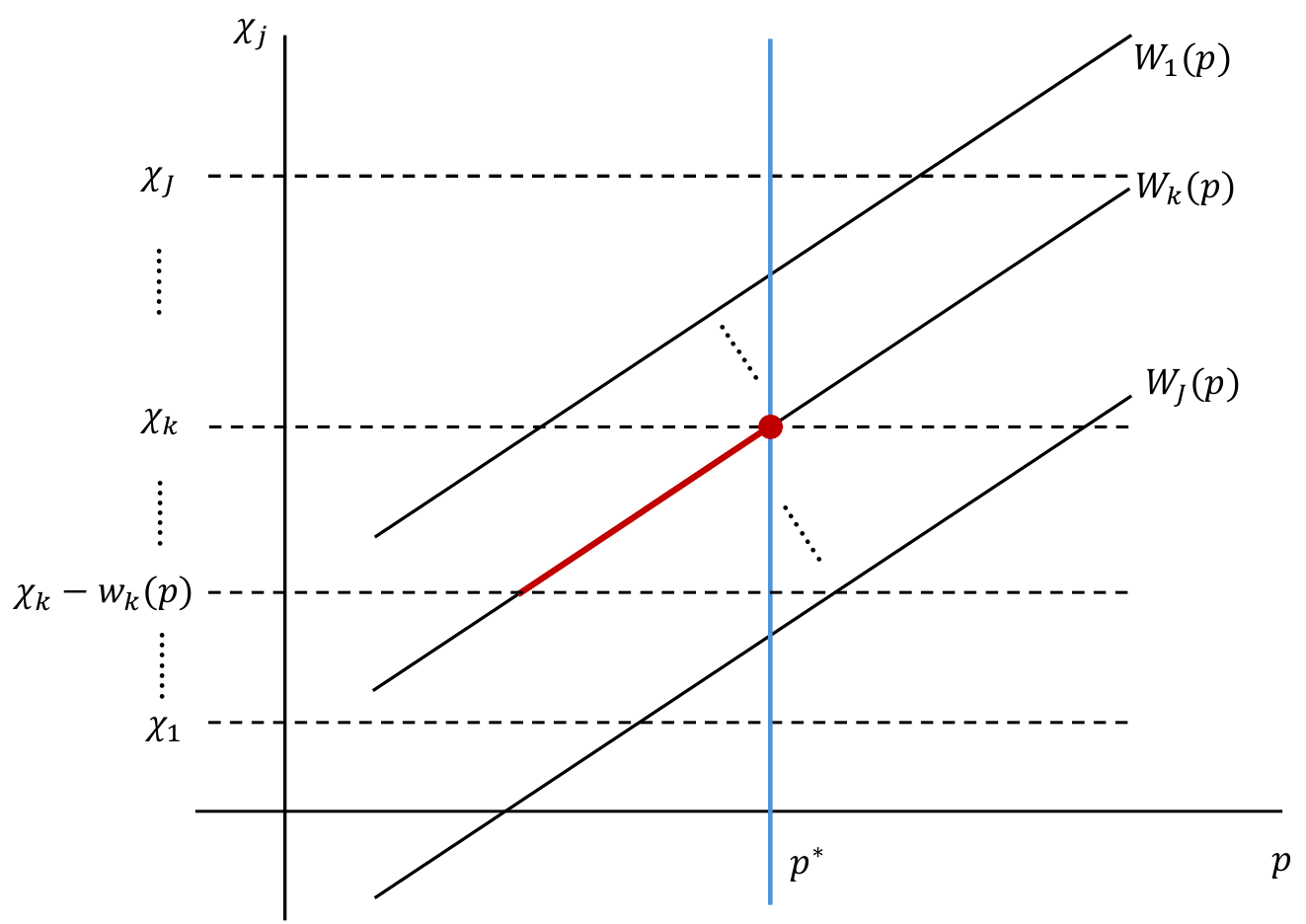}
    \caption{Purely type $k$ bypass}
    \label{fig:pure_k1}
\end{subfigure}
\hfill
\begin{subfigure}{0.325\textwidth}
    \includegraphics[width=\textwidth]{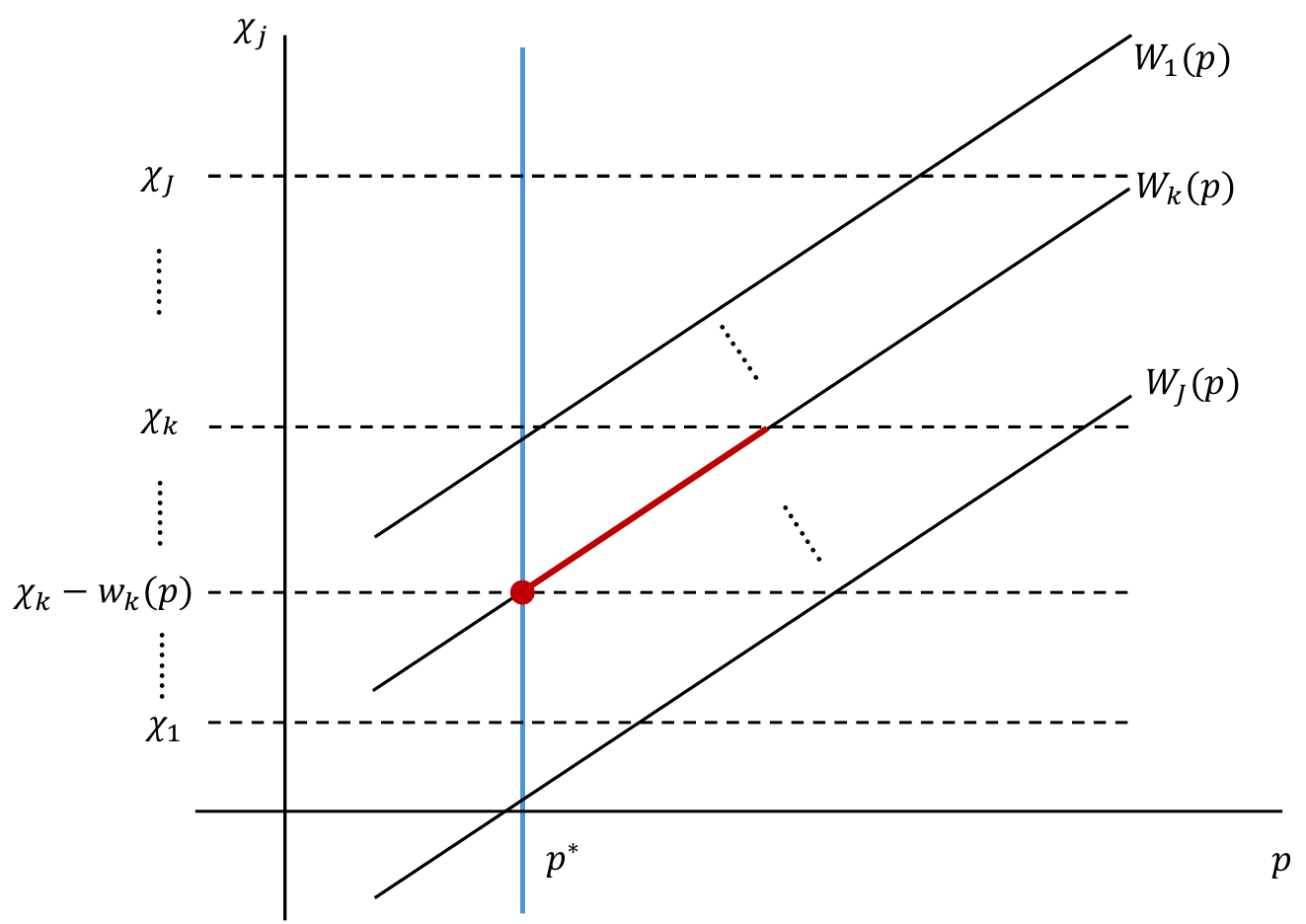}
    \caption{Purely type $k$ steadfast}
    \label{fig:pure_k2}
\end{subfigure}

\caption{These figures show type $k$'s behavior regimes.}
\label{fig:regime_k}
\end{figure*}

%Due to $0 \leq x_{1,k}^\text{s}(p) \leq w_k(p)$ , this is feasible if and only if:
%\begin{align}\label{eq:regime1}
%  W_{k}(p) < \chi_k < W_{k-1}(p).
%\end{align}

\textbf{Regime 2: Boundary Regime for Purely type $k$.} When vehicle type $k$ satisfy:
\begin{align}
  \chi_k-W_{k}(p)=0, \label{eq:regime2_1}
\end{align}
vehicle type $k$ becomes purely bypass (i.e., $x_{1,k}^s = 0$). Vehicle types with $\chi_j > \chi_k$ are fully steadfast, while vehicle types with $\chi_j \le \chi_k$ are fully bypass. This boundary corresponds to the transition from the regime where type $k$ is mixed to the regime where type $k+1$ becomes the potentially mixed type.

When vehicle type $k$ satisfy: 
\begin{align}
  \chi_k-W_{k}(p)=w_k(p),
\end{align}
vehicle type $k$ becomes purely steadfast (i.e., $x_{1,k}^s = w_k(p)$). Vehicle types with $\chi_j \ge \chi_k$ are fully steadfast, while vehicle types with $\chi_j < \chi_k$ are fully bypass. In this case, the system transitions to the regime associated with type $k-1$.

\textbf{Check the regime pairwise disjointness of $k-1$ and $k+1$.}

Based on definition of~\ref{eq:W_k(p)}, we have $W_{k-1}(p)=W_k(p)+w_k(p)$ and $W_{k+1}(p)=W_k(p)-w_{k+1}(p)$. Since $p\in I_k$, $W_k(p)<\chi_k<W_k(p)+w_k(p)$.

For $k-1$:
\begin{align}
  W_{k-1}(p)=W_k(p)+w_k(p)>\chi_k>\chi_{k-1},
\end{align}
thus there does not exist $p\in I_{k-1}$.

For $k+1$:
\begin{align}
  W_{k+1}(p)+w_{k+1}(p)=W_k(p)<\chi_k<\chi_{k+1},
\end{align}
thus there does not exist $p\in I_{k+1}$.

In conclusion, the neighboring regime will not overlap with each other.

\textbf{Check exact range of $p$ under $I_k$}.

Based on~\eqref{eq:regime1}, we can get the current range of $p$. If $A_k>0$:
\begin{align}
    \frac{\chi_k-\sum_{j\in \mathcal{H},\chi_j>\chi_k}w_h^\text{HDV}}{w_{k}^\text{CAV}+A_k} < p < \frac{\chi_k-\sum_{j\in \mathcal{H},\chi_j>\chi_k}w_h^\text{HDV}}{A_k}.
\end{align}

If $A_k<0$, only when $w_{c=k}^\text{CAV}+A_k \geq 0$, we have:
\begin{align}
  p>\frac{\chi_k-\sum_{j\in \mathcal{H},\chi_j>\chi_k}w_h^\text{HDV}}{A_k}.
\end{align}

If $A_k=0$, we have:
\begin{align}
  p>\frac{\chi_k-\sum_{j\in \mathcal{H},\chi_j>\chi_k}w_h^\text{HDV}}{w_{k}^\text{CAV}+A_k}.
\end{align}
\end{proof}

\begingroup
\small  
\bibliographystyle{apalike}
\bibliography{ref.bib}
\endgroup

\end{document}